\documentclass[journal]{IEEEtran}

\usepackage{graphicx}
\usepackage{framed}
\usepackage{subfig}
\usepackage{indentfirst}
\usepackage[]{algorithm2e}
\usepackage{booktabs}

\usepackage{amsmath, amssymb,amsthm}
\usepackage{multirow} 
\usepackage{cleveref}
\usepackage[keeplastbox]{flushend}
\usepackage{siunitx}
\crefname{figure}{Fig.}{Fig.}
\Crefname{figure}{Fig.}{Fig.}

\usepackage{pgfplots}
\pgfplotsset{compat=newest}
\usetikzlibrary{plotmarks}
\usetikzlibrary{arrows.meta}
\usepgfplotslibrary{patchplots}
\usepackage{grffile}

\renewcommand{\vec}[1]{\boldsymbol{\mathrm{#1}}} 

\let\labelindent\relax
\usepackage{enumitem}

\usepackage[font = small]{caption}

\newtheorem{prop}{Proposition}

\newtheorem{cor}[prop]{Corollary}
\newtheorem{lemma}[prop]{Lemma}

\newtheorem{remark}{Remark}

\definecolor{mycolor1}{rgb}{0,0.5,1}%
\definecolor{mycolor2}{rgb}{0.4,0.9,0.6}%
\definecolor{mycolor3}{rgb}{0.6,0.7,0.4}%
\definecolor{mycolor4}{rgb}{0.8,0.5,0.2}%
\definecolor{mycolor5}{rgb}{0.9,0.4,0.1}%
\title{\LARGE \bf  Analyzing Traffic Delay at Unmanaged Intersections}
\author{Changliu Liu and Mykel J. Kochenderfer
\thanks{C. Liu and M. Kochenderfer are with the Department of Aeronautics and Astronautics, Stanford University, CA 94305 USA (e-mail: \tt\small changliuliu, mykel@stanford.edu).}% 
}

\begin{document}

\maketitle

\begin{abstract}
At an unmanaged intersection, it is important to understand how much traffic delay may be caused as a result of microscopic vehicle interactions. Conventional traffic simulations that explicitly track these interactions are time-consuming. Prior work introduced an analytical traffic model for unmanaged intersections. The traffic delay at the intersection is modeled as an event-driven stochastic process, whose dynamics encode microscopic vehicle interactions. This paper studies the traffic delay in a two-lane intersection using the model. We perform rigorous analyses concerning the distribution of traffic delay under different scenarios. We then discuss the relationships between traffic delay and multiple factors such as traffic flow density, unevenness of traffic flows, temporal gaps between two consecutive vehicles, and the passing order. %And we point out the trade-off between fairness and efficiency, as well as between safety and efficiency.
\end{abstract}

%\tableofcontents

\section{Introduction}
Delay at intersections affect the capacity of a road network. There are many methods to analyze traffic delay at signalized intersections \cite{mathew2014signalized, xi2015approach, jiang2005traffic}. Such analyses are able to allow better traffic control to minimize delay. With the emergence of autonomous vehicles, there is a growing interest in leaving intersections unmanaged, allowing vehicles to resolve conflicts among themselves \cite{vanmiddlesworth2008replacing}. Unmanaged intersections can reduce infrastructure cost and allow for more flexible road network designs. Various vehicle policies have been proposed for distributed conflict resolution at unmanaged intersections \cite{savic2017distributed, azimi2013reliable, liu2017distributed}. 

It is important to understand how these microscopic policies affect the macroscopic transportation system. Toward the development of an efficient transportation system, we need to quantify the traffic delay generated during vehicle interactions at those intersections. 

Delay at intersections is generally evaluated using microscopic traffic simulation \cite{gora2016traffic}. Various evaluation platforms have been developed \cite{treiber2010open}, including AIMSUN \cite{barcelo1999modelling} and VISSIM \cite{Fellendorf2010}. However, it is time-consuming to obtain the micro-macro relationship by simulation. Only ``point-wise" evaluation can be performed in the sense that a single parametric change in vehicle behaviors requires new simulations. In order to gain a deeper understanding of the micro-macro relationships, an analytical model is desirable. 

In contrast with microscopic simulation models, macroscopic flow models \cite{hoogendoorn2001state} are analytical. Traffic is described by relations among aggregated values such as flow speed and density, without distinguishing its constituent parts. The major advantage of macroscopic flow models is their tractable mathematical structure with relatively few parameters to describe interactions among vehicles. However, it remains challenging to model intersections. Though intersections can be included in the flow models as boundary constraints \cite{CORTHOUT2012343, flotterod2011operational}, it is difficult to model policies other than the first-in-first-out (FIFO) policy. To consider a variety of policies, the vehicles need to be treated as particles that interact with one another, which has not been captured by existing flow models. 

\begin{figure}[t]
\vspace{-20pt}
\begin{center}
\subfloat[\label{fig: intersection a}]{
\includegraphics[width=3cm]{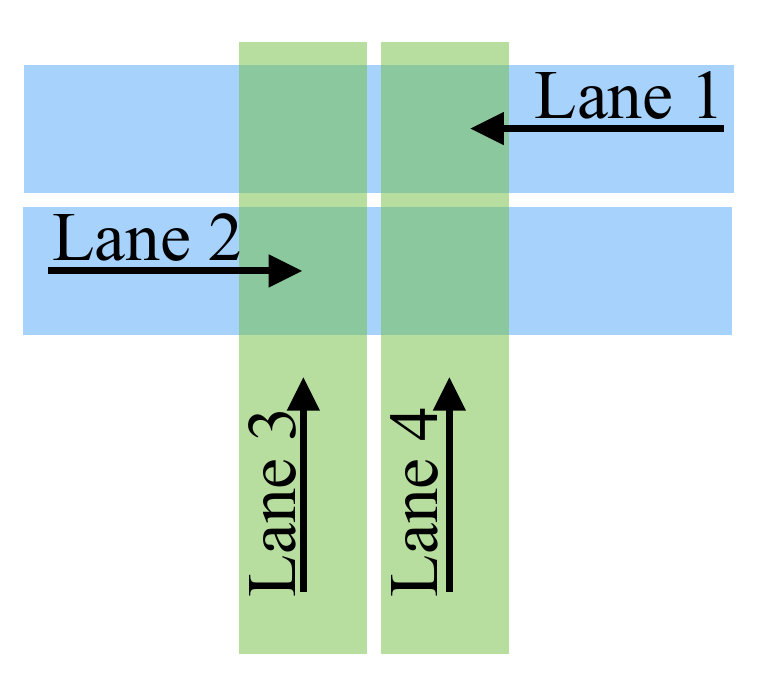}}
\subfloat[\label{fig: intersection b}]{
\includegraphics[width=5cm]{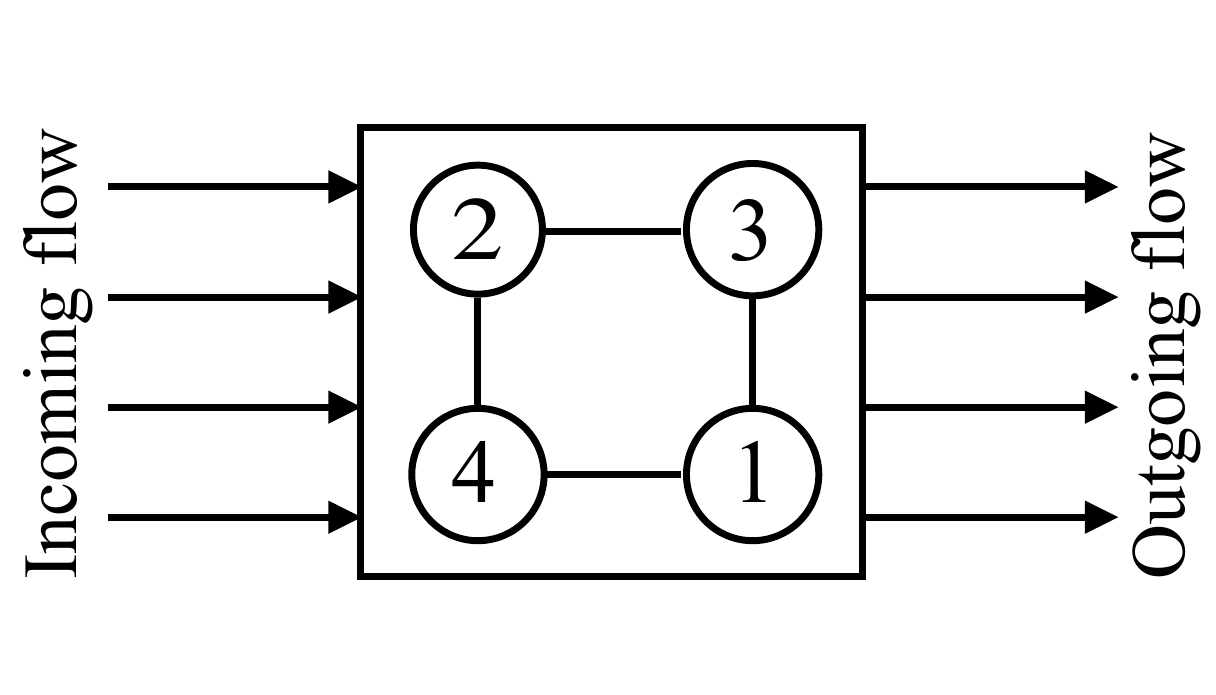}}
\caption{Intersection scenario. (a) Road topology. (b) Conflict graph.}
\label{fig: intersection}
\end{center}
\vspace{-10pt}
\end{figure}

\begin{figure}[t]
\begin{center}
\subfloat[\label{fig: time a}]{
\includegraphics[width=2.7cm]{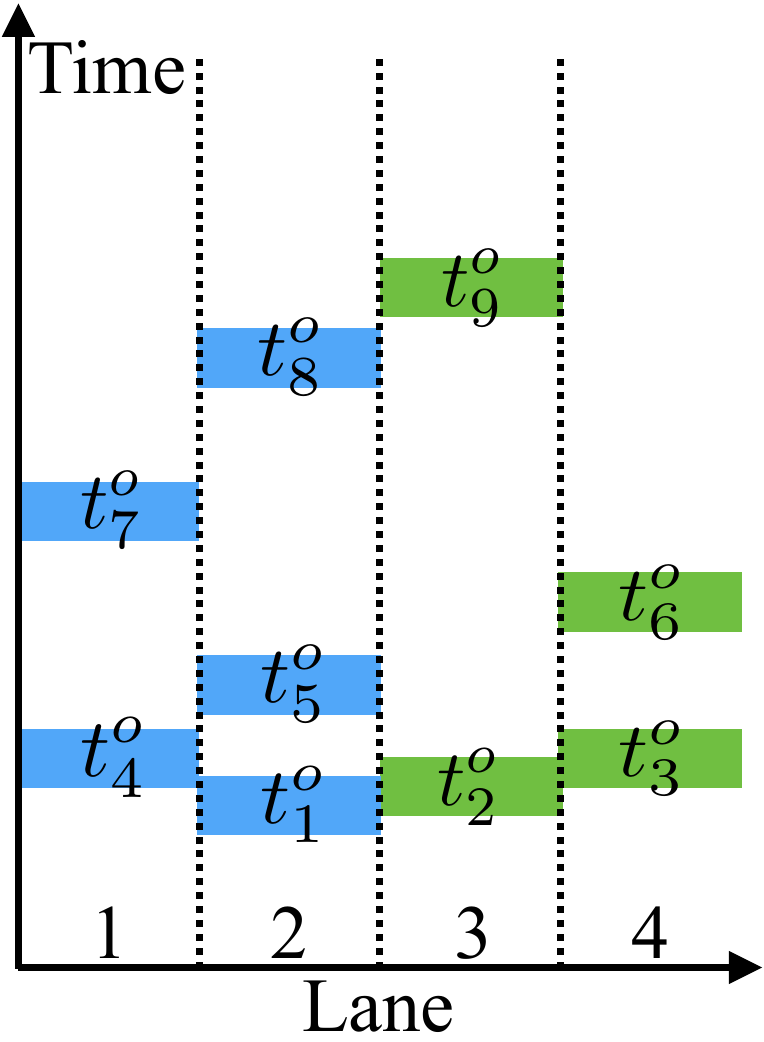}}
\subfloat[\label{fig: time b}]{
\includegraphics[width=2.7cm]{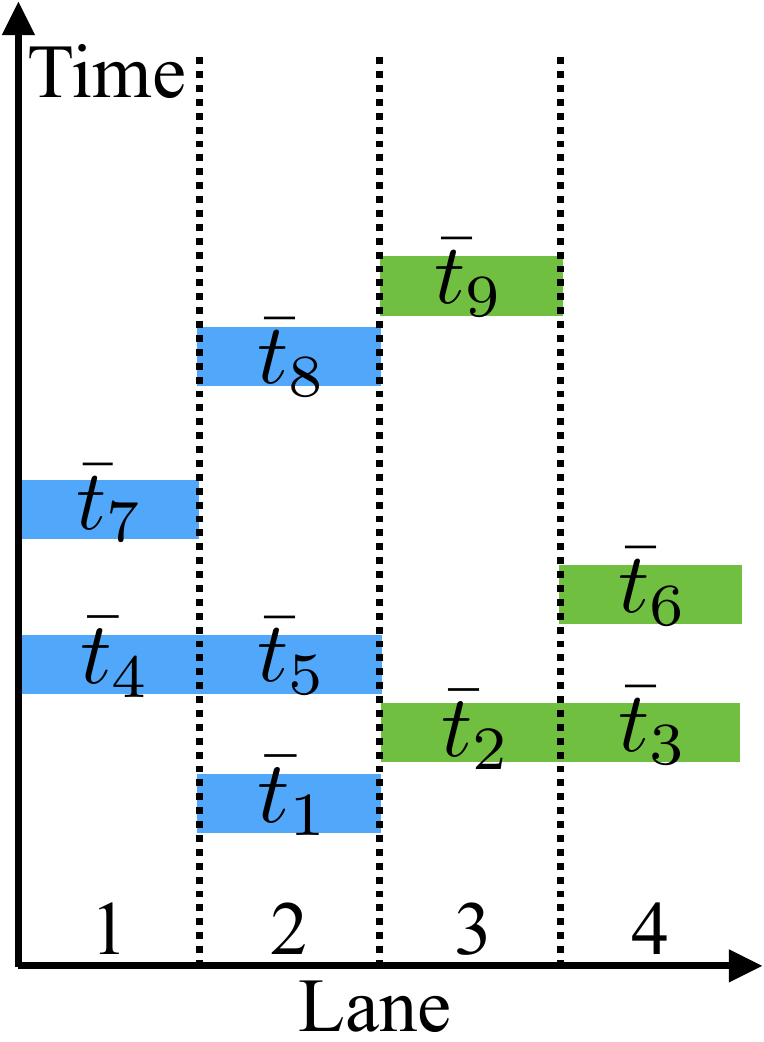}}
\subfloat[\label{fig: time c}]{
\includegraphics[width=2.7cm]{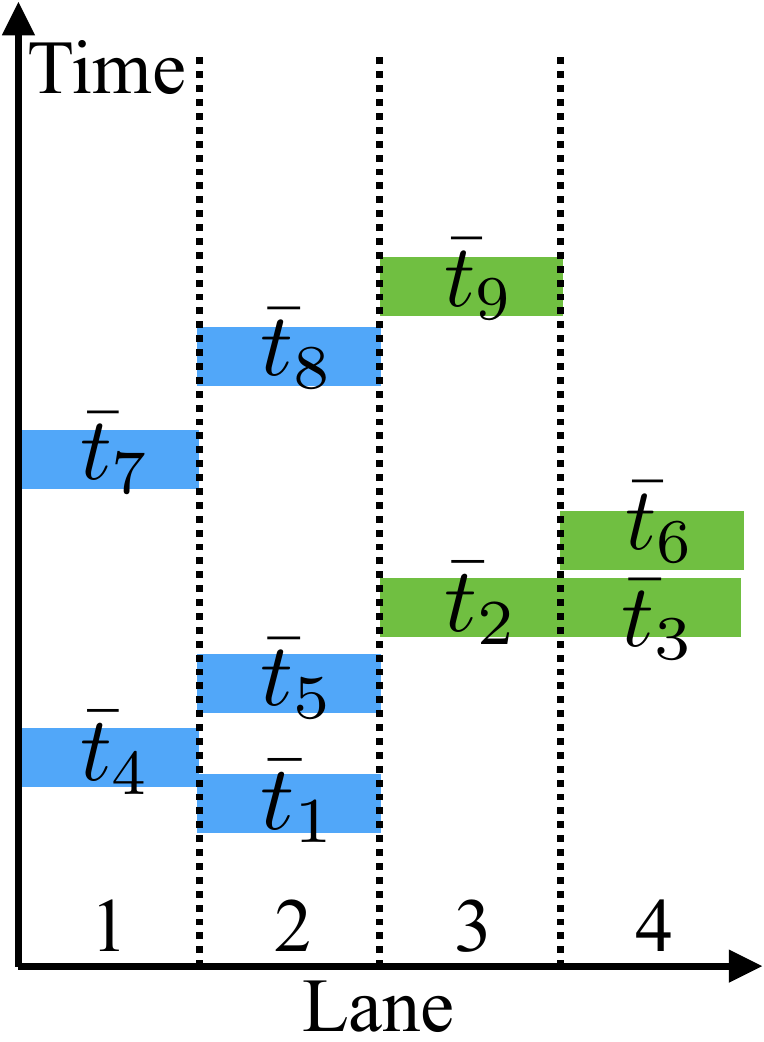}}
\caption{The time of occupancy at the intersection. (a) The desired time of occupancy. (b) The actual time of occupancy under FIFO. (c) The actual time of occupancy under FO.}
\label{fig: time}
\end{center}
\vspace{-10pt}
\end{figure}

The authors introduced an analytical traffic model \cite{liu2018analytically} to describe delays at unmanaged intersections. The model is event-driven, whose dynamics encodes equilibria resulting from microscopic vehicle interactions. It absorbs the advantages of both the microscopic simulation models and the macroscopic flow models. This paper performs detailed delay analysis at unmanaged intersections using the model. 
%The steady state distribution of delay will be derived analytically under different vehicle policies, which will then be verified against event-driven simulations \cite{liu2018analytically} and (time-driven) traffic simulations. In particular, a two-lane intersection is considered in the paper which is equivalent to lane merging case. 
The following two components in a vehicle policy strongly influence the traffic delay: 1) determination of the passing order, and 2) the temporal gap between two consecutive vehicles to pass the intersection. We will illustrate how these two components as well as the distribution of incoming traffic flows affect delay. %The temporal gap between vehicles from different directions is denoted as $\Delta_d$ and the temporal gap between vehicles from the same direction is denoted as $\Delta_s$. The gap is affected by the following factors: vehicle speed, uncertainties in perception, and etc. When the traffic density is low, we can assume $\Delta_s=0$ for simplicity.

The major contributions of this paper are:
\begin{enumerate}
\item Illustration of the usage of the analytical traffic model to obtain analytical distributions of delay.
\item Derivation of the analytical distribution of delay under two different classes of policies (i.e., two different passing orders) at a two-lane intersection;
\item Analysis of how traffic delay is affected by multiple factors at the two-lane intersection. %The factors include: the traffic flow density, unevenness of the traffic flows from different lanes, the passing order, and the temporal gap between two consecutive vehicles.
\end{enumerate}

The remainder of the paper is organized as follows. \Cref{sec: formulation} reviews the analytical traffic model and illustrates how vehicle behaviors are encoded in the model. \Cref{sec: policy} derives the analytical distribution of delay under two different classes of policies. \Cref{sec: result} shows how the traffic delay is affected by multiple factors using the analytical distribution. %\Cref{sec: discussion} discusses potential extension of the method. 
\Cref{sec: conclusion} concludes the paper.

\section{Traffic Model\label{sec: formulation}}
This section reviews an event-driven stochastic model for traffic delay at intersections \cite{liu2018analytically}. %The microscopic vehicle interactions are described first. Then the notion of equilibria of the multi-vehicle system is introduced. The traffic model considers the transitions among different equilibria. The usage of the model is discussed followed by the introduction of the model.
The following discussion considers an intersection with $K$ incoming lanes. A conflict is where two incoming lanes intersect with each other. These relationships can be described in a conflict graph $\mathcal{G}$ with the nodes being the incoming lanes and the links representing conflicts. \Cref{fig: intersection a} illustrates one possible road configuration with four incoming lanes, and \Cref{fig: intersection b} shows the resulting conflict graph. 

\subsection{Microscopic Interactions\label{sec: vehicle}}
It is assumed that the vehicles at intersections have fixed paths. To respond to others during interactions, the vehicles only change their speed profiles to adjust the time to pass the intersection \cite{altche2016time, qian2017autonomous}. This paper reduces the high dimensional speed profile for vehicle $i$ to a single state $t_i$, which denotes the time for that vehicle to pass the center of the intersection. As the mapping from $t_i$ to the speed profile is surjective, we can analyze interactions using $t_i$'s. 
The desired traffic-free time for vehicle $i$ to pass the intersection is denoted $t_i^o$. % the earliest time to pass the intersection given maximum acceleration $t_{i}^{min}$, and the latest time given maximum deceleration $t_i^{max}$. For any $t\in[t_i^{min},t_i^{max}]$, there exists a speed profile such that the vehicle pass the intersection at $t_i=t$. Hence, the mapping from $t_i$ to speed profile is surjective. So we can analyze interactions using $t_i$'s. 
The vehicles are indexed according to the desired passing time such that $t_i^o\leq t_{i+1}^o$ for all $i$.

%\begin{figure}[t]
%\begin{center}
%\includegraphics[width=4cm]{speed_profile}
%\caption{Speed profiles and state $t_i$}
%\label{fig: speed profile}
%\end{center}
%\end{figure}

At time step $k$, vehicle $i$ decides its passing time based on its desired time $t_i^o$ and its observation of others' passing times at the last time step $t_{-i}(k-1) := [t_1(k-1), \ldots, t_{i-1}(k-1), t_{i+1}(k-1), \ldots]$. The policy of vehicle $i$ is denoted
\begin{eqnarray}
t_i(k) = f(t_i^o, t_{-i}(k-1))\text{.}
\end{eqnarray}
It is assumed that all vehicles use the same policy $f$.

\subsection{Equilibria}
The equilibrium among the first $i$ vehicles is denoted $(\bar t_1^{(i)}, \ldots, \bar t_i^{(i)})$. In an equilibrium, no vehicle is willing to change the passing time before the arrival of the $(i+1)$th vehicle. Hence, the equilibrium is time-invariant, i.e.,
\begin{eqnarray}
\bar t_j^{(i)} = f(t_j^o,\bar t_{-j}^{(i)}),\forall j\leq i\text{.}\label{eq: i equilibrium}
\end{eqnarray}
It is assumed that an equilibrium can be achieved in negligible time. Hence, the system moves from the $i$th equilibrium to the $(i+1)$th equilibrium when the $(i+1)$th vehicle is included. The projected passing time for a vehicle may change from one equilibrium to another equilibrium, but will eventually converge to the actual passing time. The actual passing time $\bar t_i$ for vehicle $i$ is
\begin{eqnarray}
\bar t_i = \lim_{j\rightarrow\infty} \bar t_i^{(j)}\text{.}
\end{eqnarray}

The problem of interest is to quantify the average delay
\begin{eqnarray}
\bar d  =\lim_{N\rightarrow\infty} \frac{1}{N}\sum_{i=1}^N (\bar t_i-t_i^o)= \lim_{N\rightarrow\infty} \frac{1}{N}\sum_{i=1}^N (\bar t_i^{(N)}-t_i^o)\text{.}\label{eq: micro delay}
\end{eqnarray}

\Cref{fig: time a} illustrates the desired time of occupancy for vehicles from the four lanes in \Cref{fig: intersection a}. The bars represent the moments that the intersection is occupied by vehicles, which is centered at $t_i^o$. According to the conflict graph, the scenario in \Cref{fig: time a} is infeasible as vehicles 1, 2, 3, and 4 cannot occupy the intersection at the same time. After some negotiation and adaptation among vehicles, the actual time of occupancy becomes as shown in \Cref{fig: time b} or \Cref{fig: time c}. For an unmanaged intersection, the actual time of occupancy depends on the policies that the vehicles adopt.  \Cref{fig: time b} and \Cref{fig: time c} are different as they correspond to different policies, which will be discussed in detail in \Cref{sec: two policy}. This paper quantifies the effectiveness of the policies based on the resulted average delay.

\subsection{Traffic Model at Intersections}
For quantitative analysis, the traffic is modeled as an event-driven stochastic system with the state being the traffic delay and the input being the incoming traffic flow. The delay for lane $k$ considering $i$ vehicles is denoted $T^k_i$, which captures the difference between the passing time in the $i$th equilibrium and the traffic-free passing time of those vehicles, i.e.,
\begin{eqnarray}
T^k_i = \max_{s_{j}=k, j\leq i} \bar t_j^{(i)} - t_{i}^o\text{,}\label{eq: defn T}
\end{eqnarray}
where $s_j$ is the lane number of vehicle $j$. The input to the traffic model is the random arrival interval $x_i = t_{i+1}^o-t_i^o$ between vehicle $i+1$ and vehicle $i$, and the lane number $s_{i+1}$ of vehicle $i+1$. Define $\mathbf{T}_i:=[T^1_i, \ldots, T^K_i]^T$. The dynamics of the traffic delay follow from
\begin{eqnarray}
\mathbf{T}_{i+1} &=& \mathcal{F}(\mathbf{T}_i, x_i, s_{i+1}),\label{eq: dynamic}
\end{eqnarray}
where the function $\mathcal{F}$ depends on the policy $f$ in \eqref{eq: i equilibrium} and the road topology defined by the conflict graph $\mathcal{G}$ in \Cref{fig: intersection b}. 

It is assumed that the desired passing time of the incoming traffic flow from lane $k$ follows a Poisson distribution with parameter $\lambda_k$. The traffic flows from different lanes are independent of each other. Since the combination of multiple independent Poisson processes is a Poisson process \cite{gardiner2009stochastic}, the incoming traffic from all lanes can be described as one Poisson process $(t_1^o,t_2^o,\ldots)$ with parameter $\lambda = \sum_k \lambda_k$. The probability density for $x_i =x$ is $p_x(x) = \lambda e^{-\lambda x}$. The probability for $s_{i+1} = k$ is $P_s(k) = \frac{\lambda_k}{\lambda}$.

Given \eqref{eq: dynamic}, the conditional probability density of $\mathbf{T}_{i+1}$ given $\mathbf{T}_i$, $x_i$ and $s_{i+1}$ is
\begin{eqnarray}
p_{\mathbf{T}_{i+1}}(\mathbf{t}\mid\mathbf{T}_i, x_i, s_{i+1}) = \delta(\mathbf{t}= \mathcal{F}(\mathbf{T}_i, x_i, s_{i+1}))\text{,}
\end{eqnarray}
where $\delta(\cdot)$ is the Dirac delta function. The total distribution is
\small\begin{eqnarray}
&&p_{\mathbf{T}_{i+1}}(\mathbf{t}) \nonumber\\
&=& \sum_k P_{s}(k) \int_{x} \int_{\vec{\tau}}p_{\mathbf{T}_{i+1}}(\mathbf{t}\mid\vec\tau, x, k)p_{\mathbf{T}_i}(\vec{\tau})d\vec{\tau} p_x(x)dx\nonumber\\
&=& \sum_k P_{s}(k) \int_{\mathcal{F}(\vec\tau, x, k)=\mathbf{t}} \delta(0)p_{\mathbf{T}_i}(\vec\tau)p_{x}(x)d\vec\tau dx\text{,}\label{eq: probability}%\\&=& \sum_k P_{s}(k) \int_0^{\infty} \int_{\tau = \mathcal{F}^{-1}(\mathbf{t}, x, k)} \delta(0)p_{\mathbf{T}_i}(\tau)p_{x}(x)d\tau dx
\end{eqnarray}\normalsize
which involves integration over a manifold. %It is easy to show that the function sequence $\{p_{\mathbf{T}_{i}}\}_i$ lies on the unit circle in $L_1$. 
The cumulative probability of $\mathbf{T}_i$ is denoted  $P_{\mathbf{T}_i}(\vec{t})= \int_{-\infty}^{(t^1)^+}\cdots\int_{-\infty}^{(t^k)^+}p_{\mathbf{T}_i}(\tau^1,\ldots,\tau^k) d\tau^1\ldots d\tau^k$ for $\vec{t}=[t^1,\ldots,t^k]$. %For simplicity, we also  write $\int_{-\infty}^{t} p_{\mathbf{T}_i}(\tau) d\tau$. 

In this paper, we investigate the steady state distribution $p_{\mathbf{T}}:=\lim_{i\rightarrow\infty}p_{\mathbf{T}_{i}}$. Necessary conditions for the convergence of $\lim_{i\rightarrow\infty}p_{\mathbf{T}_{i}}$ are provided in \Cref{sec: policy}. 
For simplicity, define the functional mapping $\mathcal{M}$ as
\begin{equation}
\mathcal{M}(p)(\mathbf{t}) =   \sum_k P_{s}(k)\int_{\mathcal{F}(\vec\tau,x,k)=\mathbf{t}} \delta(0) p_{x}(x)p(\vec\tau) dxd\vec\tau \text{.}\label{eq: mapping m}
\end{equation}
The steady state distribution $p_{\mathbf{T}}$ is a fixed point under $\mathcal{M}$.
%Note that $\mathcal{M}$ is a norm-preserving mapping on the unit circle in the function space $L_1$.

\subsection{Usage of the Model}
Under the model, the distribution of vehicle delay can either be obtained through direct analysis or  event-driven simulation.

\subsubsection{Theoretical Analysis}\label{sec: method analysis}
The vehicle delay introduced by the $(i+1)$th vehicle is
\begin{eqnarray}
d_{i+1} = \sum_{j\leq i} \left(\bar t_j^{(i+1)} - \bar t_j^{(i)}\right)+\bar t_{i+1}^{(i+1)} - t_{i+1}^{*}\text{.}\label{eq: delay}
\end{eqnarray}
In the case that the introduction of a new vehicle only affects the last vehicle in other lanes (which is usually the case),
\begin{eqnarray}
d_{i+1} = T_{i+1}^{s_{i+1}} + \sum_{k\neq s_{i+1}} (T_{i+1}^k - T_{i}^k + x_i)\text{.} \label{eq: scalar delay and traffic delay}
\end{eqnarray}

Hence, to obtain an analytical steady state distribution of vehicle delay, we need to 1) obtain \eqref{eq: dynamic} from microscopic interactions models, then 2) solve the fixed point problem $\mathcal{M}(p)=p$ for the steady state distribution $p_{\mathbf{T}}$, and finally 3) compute the steady state distribution of vehicle delay $p_d$ from $p_{\mathbf{T}}$ by \eqref{eq: scalar delay and traffic delay}. \Cref{sec: policy} illustrates the procedures for the derivation. 

The relationship between $\bar d$ in \eqref{eq: micro delay} and $d_i$ in \eqref{eq: delay} is
\begin{eqnarray}
\bar d =  \lim_{N\rightarrow\infty} \frac{1}{N}\sum_{i}d_{i}\text{.}
\end{eqnarray}
According to the central limit theorem, the system is ergodic such that the average delay of all vehicles equals the expected delay introduced by a new vehicle (moving from one equilibrium to another equilibrium) in the steady state,
\begin{eqnarray}
E(\bar d) = \lim_{i\rightarrow\infty} E(d_{i})\text{.}
\end{eqnarray}

\subsubsection{Event-Driven Simulation (EDS)}
The transition of the distribution from one equilibrium to another can also be simulated. Unlike conventional time-driven traffic simulation, we can perform event-driven simulation, which is more efficient. Many particles need to be generated for $\mathbf{T}_0$, each corresponding to one traffic scenario. Those particles are then propagated according to \eqref{eq: dynamic} by randomly sampling $x_i$ and $s_{i+1}$. As the particles propagate, either the distribution diverges or we obtain the steady state distribution of delay.% $p_{\mathbf{T}}$ and vehicle delay $p_d$.

\section{Steady State Distribution of Delay\label{sec: policy}}

This section derives the steady state distribution of delay under two classes of frequently used policies in a two-lane intersection using the method discussed in \Cref{sec: method analysis}. The two policies are the first-in-first-out (FIFO) policy and the flexible order (FO) policy, which entail different passing orders. The required temporal gap between vehicles from different directions is denoted $\Delta_d$. The required temporal gap between vehicles from the same direction is denoted $\Delta_s$. The gap is affected by the following factors: vehicle speed, uncertainties in perception, and etc. %When the traffic density is low, we may assume $\Delta_s=0$ for simplicity.

\subsection{Vehicle Policies\label{sec: two policy}}
%In this paper, we investigate two vehicle policies such that the vehicles are only allowed to slow down when new vehicles are added, i.e., $\bar t_j^{(i+1)}\geq \bar t_j^{(i)}$ for all $i\geq j$. In particular, the following two policies will be studied:
The two classes of policies correspond to two ways to determine the passing order.

\subsubsection{FIFO}
The passing order is solely determined according to the arrival time (which is taken to be the desired passing time $t_i^o$). The actual passing time for vehicle $i$ should be after the actual passing times for all conflicting vehicles $j$ such that $j<i$.\footnote{Some authors define FIFO to be such that vehicle $i$ should yield to vehicle $j$ for all $j<i$ no matter there is a conflict or not. The FIFO strategy presented in this paper is similar to the Maximum Progression Intersection Protocol (MP-IP) \cite{azimi2013reliable}. Nonetheless, there is no difference between the two in the two-lane scenario.} As the passing order is fixed, the actual passing time will not be affected by later vehicles, i.e., $\bar t_j = \bar t_j^{(i)} = \bar t_j^{(j)}$ for all $i> j$. For vehicle $i$,
\begin{equation}
\bar t_i^{(i)} := \max \{t_i^o, \mathcal{D}_i, \mathcal{S}_i\}\text{,}\label{eq: fifo micro}
\end{equation}
where $\mathcal{D}_i$ is the earliest passing time considering vehicles from other lanes, and $\mathcal{S}_i$ is the earliest passing time considering vehicles from the ego lane.
\begin{subequations}
\begin{align}
\mathcal{D}_i &= \max_{j}(\bar t_j^{(i)}+\Delta_d)\text{ s.t. } j<i, (s_{j},s_i)\in \mathcal{G}\text{,}\\
\mathcal{S}_i &= \max_{j}(\bar t_j^{(i)}+\Delta_s)\text{ s.t. } j<i, s_{j}=s_i\text{.}
\end{align}
\end{subequations}
The effect of FIFO is illustrated in \Cref{fig: time b}.

%\subsubsection{Maximum progression intersection protocol (MP-IP)}
%
%The passing order is determined according to the arrival time, but vehicles with low priority can pass the intersection if there is no conflict with high priority vehicles. This strategy can be achieved if vehicles know others' intentions through communication or better prediction models. Hence, 
%\begin{eqnarray}
%\bar t_i^{(i)} = \max (t_i^o, \max_{j<i, (s_{j},s_i)\in \mathcal{G} \text{ or } s_j=s_i}(\bar t_j^{(i-1)}+\delta t_{i,j}))
%\end{eqnarray}
%where $\delta t_{i,j}$ follows from \eqref{eq: delta t}. The actual passing time will also not be affected by succeeding vehicles, i.e., $\bar t_j^{(i)} = \bar t_j^{(j)}$ for all $i> j$.
%
%\subsubsection{Advanced maximum progression intersection protocol (AMP-IP)}
%
%This strategy inherits from MP-IP. In addition to MP-IP, the lower priority vehicles are allowed to cross and clear the conflict zone before the earliest possible arrival of the higher-priority vehicle to that conflict zone. Hence
%\begin{eqnarray}
%\bar t_i^{(i)} &=& \min_t \{t\geq t_i^o\}\bigcap \left(\cap_{j<i, s_j=s_i}\{t\geq \bar t_j^{(i-1)}+\Delta_s\}\right)\nonumber\\
%&& \bigcap \left(\cap_{j<i, (s_j,s_i)\in\mathcal{G}}\{|t-\bar t_j^{(i-1)}|\geq\Delta_d\} \right)
%\end{eqnarray}
%The actual passing time will also not be affected by succeeding vehicles, i.e., $\bar t_j^{(i)} = \bar t_j^{(j)}$ for all $i> j$.

\subsubsection{FO}

\begin{figure}[t]
\begin{center}
\subfloat[Domain]{
\includegraphics[width=4.3cm]{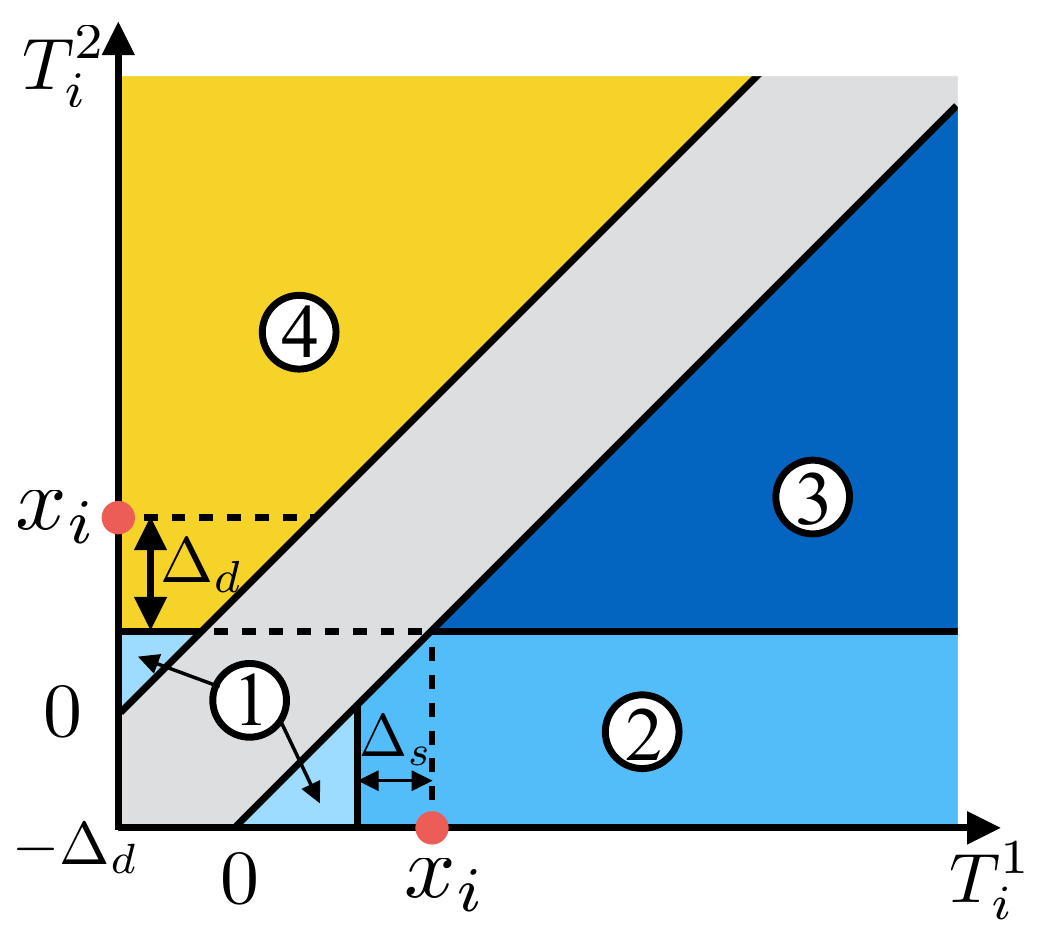}}
\subfloat[Value]{
\includegraphics[width=4.2cm]{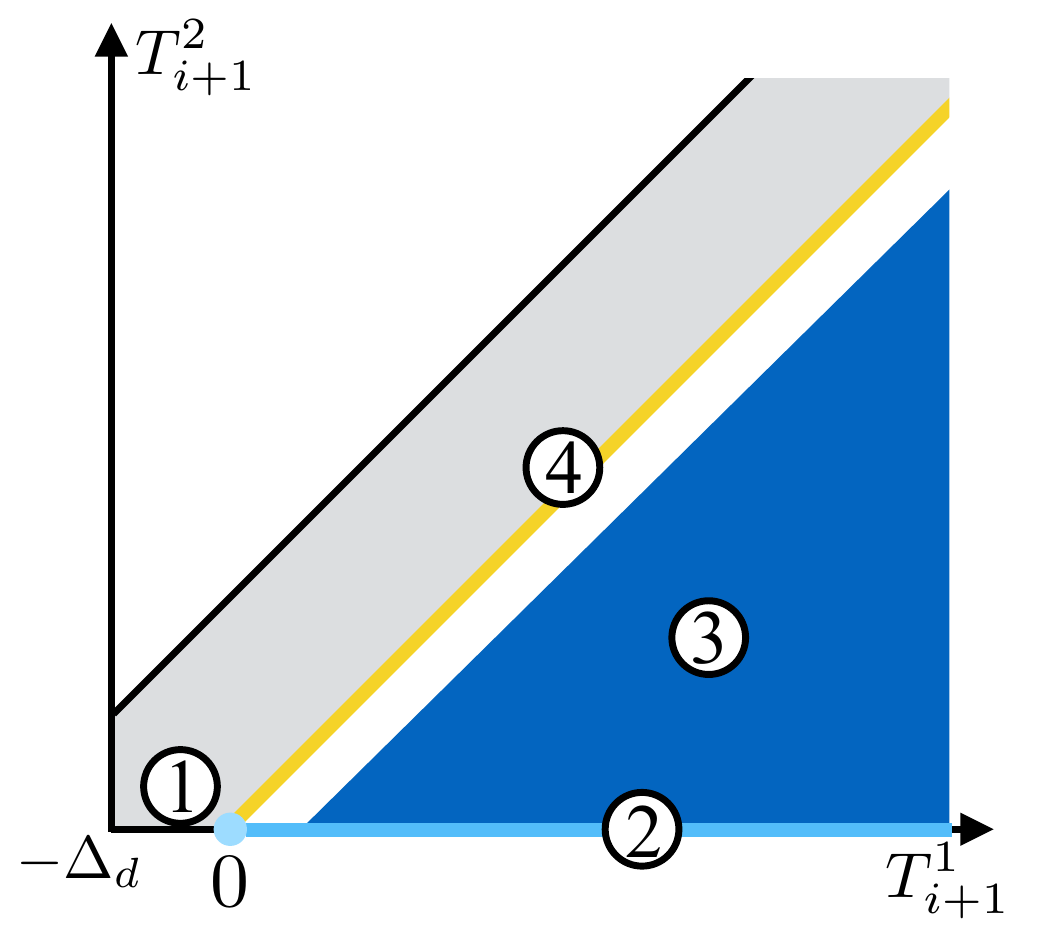}}
\caption{Illustration of the mapping \eqref{eq: dynamic} under FIFO for $s_{i+1} = 1$.}
\label{fig: FIFO mapping}
\end{center}
\vspace{-10pt}
\end{figure}

This strategy allows high priority vehicles to yield to low priority vehicles if low priority vehicles can arrive earlier. The passing order may change over time.  At step $i$, let $\bar t_i^{(i-1)} := \max \{t_i^o, \max_{j<i, s_j=s_i}(\bar t_j^{(i-1)}+\Delta_s)\}$ be the earliest possible time for vehicle $i$ to pass considering its front vehicles in the ego lane. Sort the list $(\bar t_1^{(i-1)},\ldots,\bar t_{i-1}^{(i-1)},\bar t_i^{(i-1)})$ in ascending order and record the ranking in $Q:\mathbb{N}\rightarrow\mathbb{N}$. If there is a tie, the vehicle with a smaller index is given a smaller $Q$ value. For the first vehicle in $Q$, i.e., vehicle $k = Q^{-1}(1)$, the passing time is $\bar t_k^{(i)} : = \bar t_k^{(i-1)}$. By induction, assuming that $\bar t_j^{(i)}$ for $Q(j)<Q(k)$ has been computed, then
\begin{equation}
\bar t_k^{(i)} : = \max \{\bar t_k^{(i-1)}, \mathcal{D}_k^i,\mathcal{S}_k^i\}\text{,}\label{eq: fo micro}
\end{equation}
where 
\begin{subequations}
\begin{align}
\mathcal{D}_k^i &= \max_{j}(\bar t_j^{(i)}+\Delta_d)\text{ s.t. } Q(j)<Q(k), (s_{j},s_k)\in \mathcal{G}\text{,}\\
\mathcal{S}_k^i &= \max_{j}(\bar t_j^{(i)}+\Delta_s)\text{ s.t. } Q(j)<Q(k), s_{j}=s_k\text{.}
\end{align}
\end{subequations}

Under FO, the actual passing time may change over time. There is a distributed algorithm \cite{liu2017distributed} to implement this policy where the vehicles do not necessarily need to compute the global passing order. The effect of FO is illustrated in \cref{fig: time c}. Vehicles in the same direction tend to form groups and pass together. For a two-lane intersection, the passing order is changed if and only if the next vehicle can pass the intersection earlier than the last vehicle in the other lane.

%\begin{eqnarray}
%\bar t_i^{(i)} &= \min_t \{t\geq t_i^o\}\bigcap \left(\cap_{j<i, s_j=s_i}\{t\geq \bar t_j^{(i-1)}+\Delta_s\}\right)\nonumber\\
%& \bigcap \left(\cap_{j<i, (s_j,s_i)\in\mathcal{G}}\{|t-\bar t_j^{(i-1)}-\frac{\Delta_d}{2}|\geq\frac{\Delta_d}{2}\} \right)\\
%\bar t_j^{(i)} &=
%\end{eqnarray}

\subsection{Case 1: Delay under FIFO}
Following from \eqref{eq: defn T} and \eqref{eq: fifo micro}, the dynamic equation \eqref{eq: dynamic} for FIFO can be computed, which is listed in \Cref{table: FIFO mapping} and illustrated in \Cref{fig: FIFO mapping}. Only the case for $s_{i+1}=1$ is shown. Define a conjugate operation $(\cdot)^*$ as $i^* := 3-i$. The case for $s_{i+1}=2$ can be obtained by taking the conjugate of all superscripts. In order to bound the domain from below, let $T_i^{j} = \max\{T_i^{j},-\Delta_d\}$ for all $i$ and $j \in\{ 1,2\}$. There are four smooth components in the mapping as illustrated in \Cref{fig: FIFO mapping} and \Cref{table: FIFO mapping}. Region~1 corresponds to the case that there is enough gap in both lanes for vehicle $i+1$ to pass without any delay. Regions~2 and 3 correspond to the case that the last vehicle is from the ego lane and it causes delay for vehicle $i+1$. Region~4 corresponds to the case that the last vehicle is from the other lane and it causes delay for vehicle $i+1$.
%in the two-lane case is
%\begin{eqnarray}
%T_{i+1}^{s_{i+1}} &=&  \max\{0,T_{i}^{s_{i+1}}+\Delta_s-x_i,T_{i}^{3-s_{i+1}}+\Delta_d-x_i\} \\
%T_{i+1}^{3-s_{i+1}} &=&  \max\{T_{i}^{3-s_{i+1}} - x_i,-\Delta_d\}
%\end{eqnarray}
%Then $T_{i+1}^{s_{i+1}} \geq T_{i+1}^{3-s_{i+1}} - \Delta_d$. The mapping is illustrated in \Cref{fig: FIFO mapping} and in \Cref{table: FIFO mapping} for $s_{i+1} = 1$. There are four components as illustrated in the figure and in the table. The case for $s_{i+1}=2$ is symmetric.

\begin{table}[t]
\caption{The mapping \eqref{eq: dynamic} under FIFO for $s_{i+1} = 1$.}
\vspace{-10pt}
\begin{center}
\begin{tabular}{ccc}
\toprule
Region & Condition & Value\\
\midrule
1 & $\begin{array}{c}T_i^1<x_i-\Delta_s \\ T_i^2<x_i-\Delta_d\end{array}$ & $\begin{array}{cc} T_{i+1}^1 = 0 \\ T_{i+1}^2 = -\Delta_d\end{array}$\\
\midrule
2 & $\begin{array}{c}T_i^1\geq x_i-\Delta_s \\ T_i^2<x_i-\Delta_d \\ T_i^2<T_i^1\end{array}$ & $\begin{array}{cc} T_{i+1}^1 = T_i^1+\Delta_s-x_i \\ T_{i+1}^2 = -\Delta_d\end{array}$\\
\midrule
3 & $\begin{array}{c}T_i^2\geq x_i-\Delta_d \\ T_i^2<T_i^1\end{array}$ & $\begin{array}{cc} T_{i+1}^1 = T_i^1+\Delta_s-x_i \\ T_{i+1}^2 = T_i^2-x_i\end{array}$\\
\midrule
4 & $\begin{array}{c}T_i^2\geq x_i-\Delta_d \\ T_i^2>T_i^1\end{array}$ & $\begin{array}{cc} T_{i+1}^1 = T_i^2+\Delta_d-x_i \\ T_{i+1}^2 = T_i^2-x_i\end{array}$\\
\bottomrule
\end{tabular}
\end{center}
\label{table: FIFO mapping}
\end{table}%

%Without loss of generality, assume $s_{i+1}=1$, let $t^1=T_{i+1}^{1}$ and $t^2 = T_{i+1}^{2}$, then $t^1\geq t^2+\Delta_d$. $t^1=0$ and $t^2=-\Delta_d$ corresponds to 
%\begin{eqnarray}
%T_i^{1}+\Delta_s-x_i \leq 0, T_i^{2}+\Delta_d-x_i \leq 0
%\end{eqnarray}
%
%$t^1>0$ and $t^2>-\Delta_d$ corresponds to 
%\begin{eqnarray}
%\left\{\begin{array}{ll}
%T_i^{1}+\Delta_s-x_i =t^1, T_i^{2}-x_i =t^2 & \text{for }T_i^{1}>T_i^{2}\\
%T_i^{2}+\Delta_d-x_i =t^1, T_i^{2}-x_i =t^2 & \text{for }T_i^{2}>T_i^{1}
%\end{array}\right.
%\end{eqnarray}

Given the dynamic equation, the probability \eqref{eq: probability} can be computed. For simplicity, we only show the case for $t^1>t^2$. The case for $t^1<t^2$ is symmetric. 
%\begin{subequations}
%\begin{align}
%&P_s(1)\int_{0}^\infty P_{\mathbf{T}_i}(t^1+x-\Delta_s,x-\Delta_d)p_xdx \text{ for } t^2 = -\Delta_d\\
%&P_s(1)\int_{0}^\infty \int_{-\Delta_d}^{t^2+x-\Delta_d} p_{\mathbf{T}_i}(\tau, t^2+x)d\tau p_x dx \text{ for } t^1 = t^2 + \Delta_d\\
%&P_s(1)\int_{0}^\infty p_{\mathbf{T}_i}(t^1-\Delta_s+x, t^2+x)p_xdx \text{ for } t^1>t^2 +\Delta_d
%\end{align}
%\end{subequations}
When $t^2=-\Delta_d$, $P_{\mathbf{T}_{i+1}}(t^1,-\Delta_d)=$
\begin{eqnarray}
 %&=& P_s(1)\int_{0}^\infty\int_{-\Delta_d}^{x-\Delta_d}\int_{-\Delta_d}^{x-\Delta_s} \delta(0)p_{\mathbf{T}_i}(\tau^1, \tau^2)p_xd\tau^1d\tau^2dx\\
P_s(1)\int_{0}^\infty P_{\mathbf{T}_i}(t^1+x-\Delta_s,x-\Delta_d)p_xdx\text{.}\label{eq: fifo 0}
%\\P_{\mathbf{T}_{i+1}}(-\Delta_d,t^2) = P_s(2)\int_{0}^\infty P_{\mathbf{T}_i}(x-\Delta_d,t^2+x-\Delta_s) p_x dx
\end{eqnarray}
When $t^1 = t^2 + \Delta_d$, $p_{\mathbf{T}_{i+1}}(t^1,t^2) =$
\begin{eqnarray}
 P_s(1)\int_{0}^\infty \int_{-\Delta_d}^{t^2+x-\Delta_d} p_{\mathbf{T}_i}(\tau, t^2+x)d\tau p_x dx\text{.}\label{eq: fifo diagonal}
\end{eqnarray}
%When $t_2-t_1 = \Delta_d$,
%\begin{eqnarray*}
%p_{\mathbf{T}_{i+1}}(t^1,t^2) = P_s(2)\int_{0}^\infty \int_{-\Delta_d}^{t^1+x_i} p_{\mathbf{T}_i}(t^1+x_i,\tau)d\tau p_x dx
%\end{eqnarray*}
For $t^1>t^2 +\Delta_d > 0 $, $p_{\mathbf{T}_{i+1}}(t^1,t^2)=$
\begin{eqnarray}
P_s(1)\int_{0}^\infty p_{\mathbf{T}_i}(t^1-\Delta_s+x, t^2+x)p_xdx\text{.}\label{eq: FIFO t1>t2}
\end{eqnarray}
%Note that $p_{\mathbf{T}}(t^1,t^2) = 0$ if $|t^1-t^2|<\Delta_d$ and $\min\{t^1,t^2\}>-\Delta_d$. %Hence, we only need to consider the case for $|t^1-t^2|\geq\Delta_d$.
%For $t^2>t^1$ not on the boundary,
%\begin{eqnarray*}
%p_{\mathbf{T}_{i+1}}(t^1,t^2)=  P_s(2)\int_{0}^\infty p_{\mathbf{T}_i}(t^1+x_i, t^2-\Delta_s+x_i)p_xdx
%\end{eqnarray*}

%Note that the total distribution $p_{\mathbf{T}_{i+1}}(t^1,t^2)$ is symmetric if $P_s(1)=P_s(2)=0.5$. 

\begin{prop}[Necessary Condition for Convergence under FIFO]\label{prop: convergence FIFO}
The distributions $\{p_{\mathbf{T}_i}\}_i$ converges for FIFO only if the following condition holds
\begin{equation}
2\lambda_1\lambda_2\Delta_d+[\lambda_1^2+\lambda_2^2]\Delta_s \leq \lambda\text{.}\label{eq: convergence FIFO}
\end{equation}
\end{prop}
\begin{proof}
The convergence of the distribution implies the convergence of the expected delay. Hence, the minimum average departure interval between two consecutive vehicles should be smaller than the average arrival interval. For two consecutive vehicles, the probability that they are from the same lane is $P_s(1)^2+P_s(2)^2$, and the probability that they are from different lanes is $2P_s(1)P_s(2)$. Hence, the minimum average departure interval is $2P_s(1)P_s(2)\Delta_d+\left[P_s(1)^2+P_s(2)^2\right]\Delta_s$. The average arrival interval is $\frac{1}{\lambda}$. The convergence of the distribution implies 
\begin{equation}
2P_s(1)P_s(2)\Delta_d+\left[P_s(1)^2+P_s(2)^2\right]\Delta_s\leq \frac{1}{\lambda}\text{.}\label{eq: convergence FIFO proof}
\end{equation}
Condition \eqref{eq: convergence FIFO} can be obtained by rearranging \eqref{eq: convergence FIFO proof}.
\end{proof}

The proof of the sufficiency of \eqref{eq: convergence FIFO} is left as future work. In the following discussion, we investigate the steady state distribution $p_{\mathbf{T}} = \mathcal{M}(p_{\mathbf{T}})$ for  $\Delta_s>0$ and $\Delta_s=0$.

%When $\Delta_s>0$, for any $\Delta>\Delta_d$, denote a decomposition of $\Delta$ as $(n,\delta)$ where  $n$ is the maximum integer such that $\Delta -n\Delta_s\in(\Delta_d-\Delta_s,\Delta_d]$ and $\delta = \Delta -n\Delta_s$.  %When $\Delta_s = 0$, $n$ can take any value and $\delta = \Delta$.

\begin{figure}[t]
\begin{center}
%\subfloat[FIFO\label{fig: fifo multiple}]{
%\includegraphics[width=17.5cm]{FIFO-multiple}}\\
%\subfloat[FO\label{fig: fo multiple}]{
%\includegraphics[width=17.5cm]{FO-multiple}}
\subfloat[FIFO\label{fig: fifo zb}]{\includegraphics[width=4.1cm]{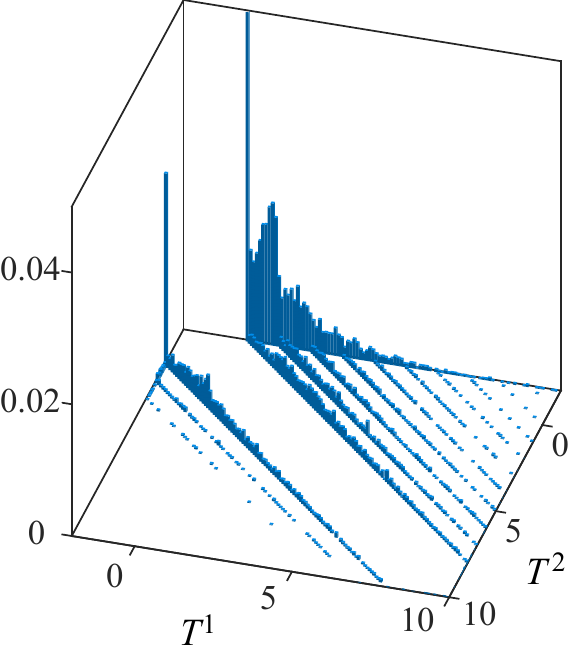}}~
\subfloat[FO\label{fig: fo zb}]{\includegraphics[width=4.1cm]{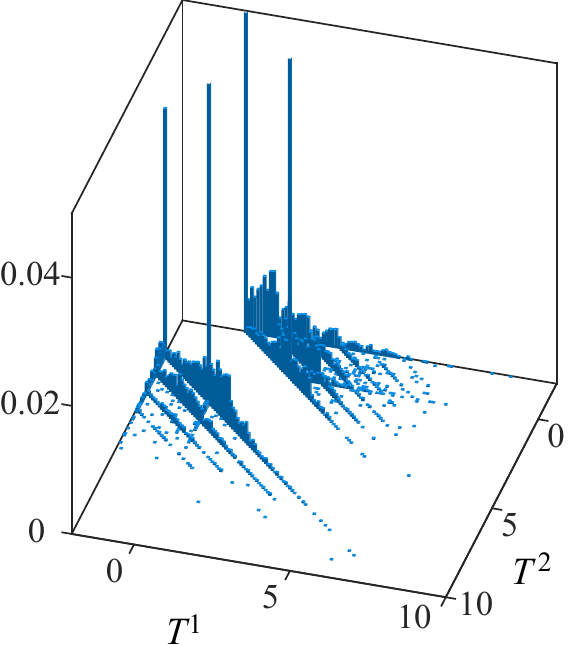}}
\caption{The steady state distribution $p_{\mathbf{T}}$ for $\lambda_1=\SI[mode=text]{0.1}{\per\second}$, $\lambda_2=\SI[mode=text]{0.5}{\per\second}$, $\Delta_d=\SI[mode=text]{2}{\second}$, and $\Delta_s=\SI[mode=text]{1}{\second}$ from EDS with $10000$ particles.}
\end{center}
\vspace{-10pt}
\end{figure}

%The propagation of $p_{\mathbf{T}_i}$ for $\lambda_1=0.1$, $\lambda_2=0.5$, $\Delta_d=2$, $\Delta_s=1$ is shown in \Cref{fig: fifo multiple} by Monte Carlo simulation. There are $10,000$ particles. At iteration 1, $P_s(1)$ percent of particles are at $(0,-\Delta_d)$, while the others are at $(-\Delta_d,0)$. The distribution approached steady state at iteration 8 with a unique pattern, which will be verified in  \Cref{prop: fifo delta_s>0}. In the following discussion, a useful lemma is presented first, followed by the analysis of the steady state distribution $p_{\mathbf{T}}$ in  \Cref{prop: fifo delta_s>0} and  \Cref{prop: fifo delta_s=0}.  

\begin{prop}[Steady State Distribution for $\Delta_s>0$ under FIFO]\label{prop: fifo delta_s>0}
When $\Delta_s >0$,  for $t>-\Delta_d$ and $\Gamma>\Delta_d$,  the following equalities hold,
\begin{subequations}
\begin{align}
p_{\mathbf{T}}(t+\Gamma, t) &= C_n^1 \int_{0}^\infty p_{\mathbf{T}}(\hat t+\gamma, \hat t)e^{-\lambda z}z^{n-1}dz\text{,}\label{eq: induction 1}\\
p_{\mathbf{T}}(t, t+\Gamma) &= C_n^2 \int_{0}^\infty p_{\mathbf{T}}(\hat t, \hat t+\gamma)e^{-\lambda z}z^{n-1}dz\text{,}\label{eq: induction 2}
\end{align}
\end{subequations}
where $n$ is the maximum integer such that $\gamma:=\Gamma -n\Delta_s\in(\Delta_d-\Delta_s,\Delta_d]$, $C_n^i =  \frac{\lambda_i^n}{(n-1)!}$, and $\hat t = t+z$. 
Moreover, $p_{\mathbf{T}}(t^1,t^2) = 0$ if $\min\{t^1,t^2\}>-\Delta_d$ and $| t^1-t^2| \neq \Delta_d + n\Delta_s$ for any $n\in\mathbb{N}$. 
\end{prop}
\begin{proof}
Since \eqref{eq: induction 1} and \eqref{eq: induction 2} are symmetric, we will only show the derivation for \eqref{eq: induction 1} for simplicity. By \eqref{eq: FIFO t1>t2},
\begin{eqnarray}
&&p_{\mathbf{T}}(t+n\Delta_s+\gamma,t) \nonumber\\
&=& P_s(1)\int_{0}^\infty p_{\mathbf{T}}(t+(n-1)\Delta_s
+\gamma+x_1, t+x_1)p_xdx_1\nonumber
\end{eqnarray}

By induction on $n$, 
\begin{equation*}
p_{\mathbf{T}}(t+\Gamma,t) = P_s(1)^{n}\int_{\vec{x}\geq\vec{0}} p_{\mathbf{T}}(t+z+\gamma, t+z)p_{\mathbf{x}}d\mathbf{x}\text{,}%\label{eq: FIFO recursive}
\end{equation*}
where $\mathbf{x} = [x_1,x_2,\ldots, x_{n}]$, $z = \sum_{k=1}^{n} x_k$, and $p_{\mathbf{x}}(\vec{x}) = \lambda^n e^{-\lambda z}$. By change of variable from $\mathbf{x}$ to $[ z ,x_2,\ldots,x_n]$,
\begin{equation*}
p_{\mathbf{T}}(t+\Gamma,t) = \lambda_1^n \int_{0}^\infty V(z,n-1) p_{\mathbf{T}}(t+z+\gamma, t+z)e^{-\lambda z}dz\text{,}
\end{equation*}
where $V(z, n-1) = \frac{1}{(n-1)!}z^{n-1}$ is the volume of an $(n-1)$-dimensional cone with depth $z$.\footnote{$V(z, n-1) = \int_0^{z}\int_0^{z-x_2}\cdots\int_0^{z-x_2-\ldots-x_{n-1}} dx_{n}\cdots dx_{3}dx_2$.} Hence, \eqref{eq: induction 1} is verified.

If $\gamma\in (\Delta_d-\Delta_s,\Delta_d)$, by definition, $p_\mathbf{T}(t+\gamma ,t)=0$. Then $p_\mathbf{T}(t+n\Delta_s+\gamma ,t)=0$ for any $n\in\mathbb{N}$ according to \eqref{eq: induction 1}. Similarly, $p_\mathbf{T}(t, t+n\Delta_s+\gamma)=0$ for any $n\in\mathbb{N}$. Hence, $p_{\mathbf{T}}(t^1,t^2) = 0$ if $\min\{t^1,t^2\}>-\Delta_d$ and $| t^1-t^2| \neq \Delta_d + n\Delta_s$ for any $n\in\mathbb{N}$.
\end{proof}

\Cref{prop: fifo delta_s>0} implies a unique ``zebra'' pattern of the steady state lane delay. This pattern is also observed in EDS shown in \cref{fig: fifo zb}. The exact solution of $p_{\mathbf{T}}$ for $\Delta_s>0$ is left as future work. In the following discussion, we derive the case for $\Delta_s=0$. The assumption that $\Delta_s=0$ is valid when the traffic density is low. \Cref{lem: zero function} is useful in the derivation of the steady state delay. 

\begin{lemma}[Zero Function]\label{lem: zero function}
For any norm-bounded $L_1$ function $f$, if $f(t) = a\int_0^\infty f(t+x)e^{-\lambda x}dx$ for all $t$ and $\lambda\leq a>0$, then $f\equiv 0$.
\end{lemma}
\begin{proof}
Multiply $e^{-\lambda t}$ on both sides, then
\begin{equation*}
e^{-\lambda t}f(t) = a\int_0^\infty f(t+x)e^{-\lambda (x+t)}dx = a\int_t^\infty f(x)e^{-\lambda x}dx\text{.}
\end{equation*}
Take derivative with respect to $t$ on both sides, then
\begin{equation*}
e^{-\lambda t}f'(t) -\lambda e^{-\lambda t}f(t) = -a f(t)e^{-\lambda t}\text{,}
\end{equation*}
which implies that $f'(t) = (\lambda-a) f(t)$ and $f(t) = C e^{(\lambda-a)t}$ for some constant $C$. However, since $\lambda-a\geq 0$, $f$ cannot be norm bounded if $C\neq 0$. Hence, $f\equiv 0$.
\end{proof}

In the following discussion, we derive the steady state distribution of delay for $\Delta_s = 0$. \Cref{prop: fifo delta_s=0} shows that when $\Delta_s=0$, the probability density is non trivial only at $p_{\mathbf{T}}(t, t-\Delta_d)$ or $p_{\mathbf{T}}(t-\Delta_d,t)$ for $t\geq 0$. Hence, we define 
\begin{equation}
g_1 (t) :=p_{\mathbf{T}}(t, t-\Delta_d), g_2 (t) :=p_{\mathbf{T}}(t-\Delta_d,t)\text{.}
\end{equation}
The function $g_i$ for $i\in\{1,2\}$ contains both finite component and delta component, denoted $\tilde g_i$ and $\widehat g_i$ respectively such that $g_i(t) = \tilde g_i (t) + \widehat g_i(t)\delta(t)$. Moreover, for $i\in\{1,2\}$, define the probability function $G_i$, value $\mathcal{M}_i$ and value $\mathcal{I}_i$ as
\begin{subequations}\label{eq: defn G M I}
\begin{align}
G_i(t) &:= \int_0^t g_i(\tau)d\tau\text{,}\\
\mathcal{M}_i &:= \int_0^{\infty} g_i(x)dx\text{,}\\
\mathcal{I}_i &:= \int_0^{\infty} g_i(x)e^{-\lambda x}dx\text{.}\label{eq: defn I}
\end{align}
\end{subequations}
Value $\mathcal{M}_i$ is the probability that lane $i$ has larger delay. 

\begin{prop}[Steady State Distribution for $\Delta_s=0$ under FIFO]\label{prop: fifo delta_s=0}
When $\Delta_s=0$, $p_{\mathbf{T}}(t^1,t^2) = 0$ if $| t^1-t^2| \neq \Delta_d$. For $i \in\{1,2\}$, the following equations hold 
\begin{eqnarray}
G_i(t) &=& \widehat g_i(0)e^{\lambda_{i^*}t}\text{, for } t\in[0,\Delta_d]\text{,}\label{eq: distribution FIFO}\\
\widehat g_i(t) &=& \left\{\begin{array}{ll}\frac{\lambda_i}{\lambda}\left[\mathcal{I}_i + e^{-\lambda\Delta_d}\mathcal{I}_{i^*}\right] & t=0\\
0 & t\neq 0\end{array}\right.\text{,}\label{eq: delta component FIFO}\\
\lambda_i\widehat g_{i^*}(0) &=& \tilde g_i(\Delta_d^-)-\tilde g_i(\Delta_d^+) \text{,}\label{eq: discontinuity FIFO}\\
\mathcal{M}_i &=& \frac{\lambda_i}{\lambda}\text{.}\label{eq: integration FIFO}
\end{eqnarray}
Moreover, when $t$ is sufficiently large,
\begin{equation}
\mathcal{M}_i-G_i(t)\propto e^{-at}\text{,}\label{eq: approx FIFO}
\end{equation}
where $a<0$ is the solution of the following equation
\begin{equation}
(a-\lambda_1)(a-\lambda_2) - \lambda_1 e^{-a\Delta_d}\lambda_2 e^{-a\Delta_d} = 0 \text{.}\label{eq: characteristic}
\end{equation}
%For $n\in\mathbb{N}$ and $t\in((n-1)\Delta_d,n\Delta_d]$,
%\begin{eqnarray}
%g_1(t) = \sum_{k=0}^{n-1} a_{n,k} t^k e^{\lambda t}, g_2(t) = \sum_{k=0}^{n-1} b_{n,k} t^k e^{\lambda t}\label{eq: fifo delta_s=0 solution}
%\end{eqnarray}
%where $a_{n,k}$ and $b_{n,k}$ are constants such that $\sum_{k=0}^{n-1}a_{n,k}^2>0$ and $\sum_{k=0}^{n-1}b_{n,k}^2>0$.
\end{prop}
\begin{proof}
We first show that $p_{\mathbf{T}}(t^1,t^2) = 0$ if $| t^1-t^2| \neq \Delta_d$. There are two cases: $\min\{t^1,t^2\}>-\Delta_d$ or $\min\{t^1,t^2\}=-\Delta_d$. 
Consider $\Gamma>\Delta_d$ and $t>-\Delta_d$. According to \eqref{eq: FIFO t1>t2},
\begin{equation}
p_{\mathbf{T}}(t+\Gamma,t) = \lambda_1 \int_0^\infty p_{\mathbf{T}}(t+x+\Gamma,t+x) e^{-\lambda x}dx\text{.}\label{eq: steady state off diagonal}
\end{equation}
By \Cref{lem: zero function}, \eqref{eq: steady state off diagonal} implies $p_{\mathbf{T}}(t+\Gamma,t)\equiv0$ for all $t>-\Delta_d$. Similarly, $p_{\mathbf{T}}(t,t+\Gamma)\equiv 0$ for all $t>-\Delta_d$. 
Moreover, for $t>0$, according to \eqref{eq: fifo 0},
\begin{equation}
p_{\mathbf{T}}(t,-\Delta_d) = \lambda_1 \int_{0}^\infty p_{\mathbf{T}}(t+x,-\Delta_d) e^{-\lambda x} dx\text{.}\label{eq: steady state boundary}
\end{equation}
By \Cref{lem: zero function}, \eqref{eq: steady state boundary} implies $p_{\mathbf{T}}(t,-\Delta_d)\equiv0$ for $t>0$. Similarly, $p_{\mathbf{T}}(-\Delta_d,t)\equiv 0$ for $t>0$. Hence, the claim is verified.

Now we compute the steady state distribution $g_i$. %In the following discussion, we show derivation for $g_1$. The derivation for $g_2$ follows easily. 
In either $g_1$ or $g_2$, there is only one point mass at $0$ by \eqref{eq: fifo 0} to \eqref{eq: FIFO t1>t2}. According to \eqref{eq: fifo 0},
\begin{eqnarray}
&&\widehat g_i(0) \nonumber\\
&=& \lambda_i \int_0^\infty \left[\int_0^x g_i(t)dt+\int_0^{x-\Delta_d} g_{i^*}(t)dt\right]e^{-\lambda x}dx\nonumber\\
&=& \lambda_i \int_0^\infty \left[\int_t^\infty e^{-\lambda x}dx g_i+ \int_0^\infty \int_{t+\Delta_d}^\infty e^{-\lambda x}dx g_{i^*}\right]dt \nonumber\\
&=& \frac{\lambda_i}{\lambda}\left[\int_0^\infty g_i(t)e^{-\lambda t}dt + \int_0^\infty g_{i^*}(t)e^{-\lambda (t+\Delta_d)}dt\right]\text{,}\label{eq: point mass FIFO delta_s=0}
\end{eqnarray}
where the second equality is obtained by changing the order of integration. By definition \eqref{eq: defn I}, \eqref{eq: point mass FIFO delta_s=0} implies \eqref{eq: delta component FIFO}.
%\begin{eqnarray}
%\widehat g_1(0) = \frac{\lambda_1}{\lambda}\left[\mathcal{I}_1 + e^{-\lambda\Delta_d}\mathcal{I}_2\right]
%\end{eqnarray}

According to \eqref{eq: fifo diagonal} and \eqref{eq: FIFO t1>t2}, for $t>0$,
\begin{equation}
\tilde g_i (t) = \lambda_i \int_{0}^\infty \left[g_i(t+x)+g_{i^*} (t+x-\Delta_d)\right]e^{-\lambda x}dx\text{,}\label{eq: g1 formula}
\end{equation}
which implies that $\tilde g_i$ is continuous except at $\Delta_d$. The discontinuity at $\Delta_d$ is caused by the point mass $\widehat g_{i^*}(0)$. By \eqref{eq: g1 formula}, the claim in \eqref{eq: discontinuity FIFO} is verified. 
By multiplying $e^{-\lambda t}$ on both sides of \eqref{eq: g1 formula} and then taking derivatives similar to the proof in  \Cref{lem: zero function}, we obtain the following differential equation
\begin{eqnarray}
\tilde g_i'(t) = \lambda_{i^*} \tilde g_i(t)-\lambda_i\tilde g_{i^*}(t-\Delta_d)\text{.}\label{eq: FIFO pde 1}
\end{eqnarray}
For $t\in(0,\Delta_d)$, since $\tilde g_{i^*}(t-\Delta_d)=0$, \eqref{eq: FIFO pde 1} implies that there exists $c_{i}\in\mathbb{R}^+$ such that
\begin{equation}
\tilde g_i(t) = c_{i}e^{\lambda_{i^*} t}\text{.}\label{eq: g1 solution FIFO}
\end{equation}
Plugging \eqref{eq: g1 solution FIFO} back to \eqref{eq: g1 formula}, the constant $c_i$ can be computed,
%\begin{equation}
%c_{1}e^{-\lambda_1t} = \lambda_1\left[\mathcal{I}_1- \widehat g_1(0) - \int_0^t c_{1}e^{-\lambda_1 x} dx + e^{-\lambda \Delta_d}\mathcal{I}_2\right]\text{.}
%\end{equation}
\begin{eqnarray}
c_{i} = \lambda_i\left[\mathcal{I}_i - \widehat g_i(0) + e^{-\lambda \Delta_d}\mathcal{I}_{i^*} \right]  = \lambda_{i^*}\widehat g_i(0)\text{.}\label{eq: c FIFO}
\end{eqnarray}
Then \eqref{eq: distribution FIFO} is verified by integrating \eqref{eq: g1 solution FIFO}. Moreover, it is easy to verify that $\mathcal{M}_i = P_s(i) = \frac{\lambda_i}{\lambda}$. Hence, \eqref{eq: distribution FIFO} to \eqref{eq: integration FIFO} are all verified.
%Integrate \eqref{eq: FIFO pde 1} on both sides from $0$ to $\infty$, $\tilde g_1(\Delta_d^-)-\tilde g_1(0^+) - \tilde g_1(\Delta_d^+)  = \lambda_2 \int_0^\infty \tilde g_1(t)dt -\lambda_1 \int_0^\infty \tilde g_2(t-\Delta_d)dt $. Hence, $-\lambda_2\widehat g_1(0) + \lambda_1\widehat g_2(0) = \lambda_2(\mathcal{M}_1-\widehat g_1(0)) - \lambda_1(\mathcal{M}_2-\widehat g_2(0))$. Then
%\begin{equation}
%\lambda_2\mathcal{M}_1 - \lambda_1\mathcal{M}_2 = 0
%\end{equation}
%Since $\mathcal{M}_1+\mathcal{M}_2 = 1$, the above equation implies \eqref{eq: integration FIFO}. 

The characteristic equation \cite{dde} of the delay differential equation \eqref{eq: FIFO pde 1} for $i\in\{1,2\}$ satisfies
\begin{equation}
\det \left(aI_2 -  \left[\begin{array}{cc}
\lambda_2 & 0\\
0 & \lambda_1\end{array}\right] + e^{-a\Delta_d}\left[\begin{array}{cc}
0 & \lambda_1\\
\lambda_2 & 0\end{array}\right] \right) = 0 \text{,}
\end{equation}
which is equivalent to the nonlinear eigenproblem \eqref{eq: characteristic}.
There are three possible solutions with $a=0$, $a>0$, and $a<0$, respectively.  Since $\lim_{t\rightarrow\infty}g_i(t)=0$, we can only take the solution $a<0$. When $t\rightarrow\infty$, $g_i(t)$ is proportional to $e^{at}$. Then \eqref{eq: approx FIFO} is verified.

\end{proof}

To compute the exact solution of the distribution, the delay differential equation (DDE) \eqref{eq: FIFO pde 1} needs to be solved. To solve the DDE, we need to compute the expression of $G_i(t)$ for $t\in((n-1)\Delta_d,n\Delta_d]$ consecutively for all $n$ considering the boundary constraints \eqref{eq: distribution FIFO} to \eqref{eq: integration FIFO}. However, as there are infinitely many segments, the complexity of the problem grows quickly. In this paper, we approximate the distribution for $t>\Delta_d$ using \eqref{eq: approx FIFO}.  By incorporating  \eqref{eq: distribution FIFO} and \eqref{eq: integration FIFO}, the approximated distribution is
\begin{equation}
G_i(t) = \left\{\begin{array}{ll}
\widehat g_i(0)e^{\lambda_{i^*}t} &  t\leq \Delta_d\\
\frac{\lambda_i}{\lambda}(1-e^{a(t-\Delta_d)})+G_i(\Delta_d)e^{a(t-\Delta_d)} & t>\Delta_d
\end{array}\right.\text{.}\label{eq: approx distribution FIFO}
\end{equation}
There is only one unknown parameter $\widehat g_i(0)$, which can be solved by the remaining equations in \Cref{prop: fifo delta_s=0}. However, the approximated distribution \eqref{eq: approx distribution FIFO} is not simultaneously compatible with \eqref{eq: delta component FIFO} and \eqref{eq: discontinuity FIFO}. We need to relax either condition. Equation \eqref{eq: delta component FIFO} is a global condition as it is related to the integral of the distribution. Equation \eqref{eq: discontinuity FIFO} is a local condition as it concerns the discontinuous point of the probability density. 

\begin{remark}[Approximation 1]
In the first approximation, the local condition \eqref{eq: discontinuity FIFO} is relaxed. Then $\widehat g_i(0)$ is obtained by solving \eqref{eq: delta component FIFO} and \eqref{eq: approx distribution FIFO}, 
\begin{equation}
\widehat g_i(0) = \frac{a  \lambda_{i}  y  \left((\lambda_i-a){\lambda_{i}}  (y^2-1)    + (a-\lambda)    y_i \left[\lambda_{i^*} +  \lambda_{i}   y \right]\right)}{B_i}\text{,} \label{eq: FIFO approx 1}
\end{equation}
where $y:=e^{-\lambda\Delta_d}$, $y_i:=e^{-\lambda_i\Delta_d}$, and
\begin{equation}
\begin{array}{rl}
B_i &= {\lambda}^2  (  a^2  y (y-y_i) (1-y_i) + a(a-\lambda)  y_i   \\&  +(a-\lambda_i)  \lambda  y^2  (y_i-1) + (2  a - \lambda)  \lambda  y  y_i (1-y_i)   \\
&+ (a-\lambda)  \lambda_{i}  y  y_i^2     + \lambda_i  \lambda_{i^{*}}  y_i + {\lambda_{i}}^2  y^2  y_i - a  \lambda_{i}  y^2)\text{.}
\end{array}
\end{equation}
\end{remark}

\begin{remark}[Approximation 2]
In the second approximation, we relax the global condition \eqref{eq: delta component FIFO}. Then $\widehat g_i(0)$ is obtained by solving \eqref{eq: discontinuity FIFO} and \eqref{eq: approx distribution FIFO}, 
\begin{eqnarray}
\widehat{g}_i(0) = \frac{a \lambda_{i} y e^{a \Delta_{d}}  \left(\lambda_{i} + \lambda_{i^*} y_i - a e^{a\Delta_{d}}\right)}{\lambda  \left(a  \lambda  e^{a  \Delta_{d}} - a^2 e^{2 a \Delta_{d}}  + \lambda_{i} \lambda_{i^*} (y-1)\right)}\text{.}\label{eq: FIFO approx 2}
\end{eqnarray}
\end{remark}

The accuracy of the two approximations against the steady state distribution obtained from EDS with $10000$ particles is shown in \cref{fig: approx FIFO}. Though both underestimate the delay, \eqref{eq: FIFO approx 1} provides a better approximation because it preserves the global property. In the following discussion and analysis, we use the first approximation.

\begin{cor}[Approximated Steady State Vehicle Delay]\label{cor: scalar delay FIFO}
When $\Delta_s=0$, under the approximation \eqref{eq: approx distribution FIFO}, the steady state vehicle delay has the distribution
\begin{equation}
P_d(t) = \left\{\begin{array}{ll}
\widehat g_1(0)e^{\lambda_2t} + \widehat g_2(0)e^{\lambda_1t} & t\leq \Delta_d\\
1-e^{a(t-\Delta_d)}+P_d(\Delta_d)e^{a(t-\Delta_d)} & t>\Delta_d
\end{array}
\right.\text{,}\label{eq: scalar delay FIFO}
\end{equation}
with expected delay
\begin{equation}
E(d) = \widehat g_1(0)\mathcal{E}(\lambda_2)+\widehat g_2(0)\mathcal{E}(\lambda_1) - \frac{\left(a  \Delta_{d} - 1\right)  \left(P_d(\Delta_d) - 1\right)}{a}  \text{,}\label{eq: expected delay FIFO}
\end{equation}
where
\begin{equation}
\mathcal{E}(\lambda_i) = \frac{1+e^{\Delta_{d}  \lambda_{i}}  \left(\Delta_{d}  \lambda_{i} - 1\right)}{{\lambda_{i}}}\text{.}\label{eq: integral 0 delta_d}
\end{equation}

\end{cor}
\begin{proof}
By \eqref{eq: scalar delay and traffic delay}, the vehicle delay in the steady state satisfies that $P_d(t) = G_1(t)+G_2(t)$. So \eqref{eq: scalar delay FIFO} follows from \eqref{eq: approx distribution FIFO}. The expected delay satisfies $E(d) = \int_0^\infty tdP_d(t)$. Let $\mathcal{E}(\lambda_i):=\int_0^{\Delta_d} t de^{\lambda_it}$. Then \eqref{eq: expected delay FIFO} and \eqref{eq: integral 0 delta_d} follow.
\end{proof}

\begin{figure}[t]
%\begin{center}
\input{FIFO-Uneven.tex}
\caption{Steady state traffic delay under FIFO. $\lambda_1=\SI[mode=text]{0.3}{\per\second}$, $\lambda_2=\SI[mode=text]{0.5}{\per\second}$, $\Delta_d = \SI[mode=text]{2}{\second}$, and $\Delta_s =\SI[mode=text]{0}{\second}$.}
\label{fig: approx FIFO}
\vspace{-10pt}
%\end{center}
\end{figure}

%\begin{remark}[Approximated Solution]
%In the following discussion, $a$ denotes the negative solution of \eqref{eq: characteristic}. Hence, for $t>\Delta_d$, and $i\in\{1,2\}$,
%\begin{eqnarray}
%\tilde g_i(t) = \tilde g_i(\Delta_d^+) e^{a(t-\Delta_d)}\text{.}
%\end{eqnarray}
%Then,
%\begin{eqnarray}
%\mathcal{I}_1 &=&  \frac{ \tilde g_1(\Delta_d^+) e^{- \lambda\Delta_{d}}}{\lambda -a} + \widehat{g}_1(0) \left(1+\frac{\lambda_{2} \left(1-e^{-\lambda_{1}\Delta_{d}} \right)}{\lambda_{1}} \right)\\
%\mathcal{I}_2 &=&  \frac{ \tilde g_2(\Delta_d^+) e^{- \lambda\Delta_{d}}}{\lambda -a} + \widehat{g}_2(0) \left(1+\frac{\lambda_{1} \left(1-e^{-\lambda_{2}\Delta_{d}} \right)}{\lambda_{2}} \right)
%\end{eqnarray}
%Using the expression for 
%\end{remark}
 
%\Cref{prop: fifo delta_s=0} implies that the domain of $p_{\mathbf{T}}$ is unbounded. We will see in Case 2 that the domain of $p_{\mathbf{T}}$ is bounded under FO strategy.

\subsection{Case 2: Delay under FO}
Following from \eqref{eq: defn T} and \eqref{eq: fo micro}, the dynamic equation \eqref{eq: dynamic} for FO can be computed, which is listed in \Cref{table: FO mapping} and illustrated in \Cref{fig: FO mapping} for $s_{i+1}=1$. There are eight smooth components in the mapping. Regions 1 to 4 are the same as in the FIFO case. Vehicle $i+1$ is the last one to pass the intersection. Regions 5 to 8 correspond to the case that vehicle $i+1$ passes the intersection before the last vehicle in the other lane. In regions 5 and 7, vehicle $i+1$ arrives earlier than the last vehicle in the other lane and there is enough gap in the ego lane. Hence, vehicle $i+1$ passes without delay, but the last vehicle in the other lane yields (with delay in region 5, without delay in region 7). Regions 6 and 8 correspond to the case that vehicle $i+1$ is delayed by the last vehicle in the ego lane but can still go before the last vehicle in the other lane. Delay is caused in the other lane in region 6.

%\begin{eqnarray}
%T_{i+1}^{s_{i+1}} &=& \left\{\begin{array}{ll} 
%y_i-x_i & y_i \leq T_{i}^{3-s_{i+1}}\\
%\max\{y_i-x_i,T_{i}^{3-s_{i+1}}+\Delta_d-x_i\} & y_i > T_{i}^{3-s_{i+1}}
%\end{array}\right. \\
%T_{i+1}^{3-s_{i+1}} &=&  \left\{\begin{array}{ll} 
%\max\{T_{i}^{3-s_{i+1}} - x_i,y_i+\Delta_d-x_i\} & y_i \leq T_{i}^{3-s_{i+1}}\\
%\max\{T_{i}^{3-s_{i+1}} - x_i,-\Delta_d\} & y_i > T_{i}^{3-s_{i+1}}
%\end{array}\right.
%\end{eqnarray}
%Note that $\mid T_{i+1}^{s_{i+1}} - T_{i+1}^{3-s_{i+1}}\mid \geq \Delta_d$. But it is not necessary that $T_{i+1}^{s_{i+1}} > T_{i+1}^{3-s_{i+1}}$. The mapping is illustrated .
Given the dynamic equation, the probability \eqref{eq: probability} can be computed. %For simplicity, we only show the case that $t^1>t^2$. When $t^2=-\Delta_d$, $P_{\mathbf{T}_{i+1}}(t^1,-\Delta_d) =$
%\small{
%\begin{eqnarray}
% P_s(1)\int_{0}^\infty P_{\mathbf{T}_i}(t^1+x-\Delta_s,x-\Delta_d)p_xdx
%\end{eqnarray}}\normalsize
%When $t^2\in(-\Delta_d,0)$, $t_1>t_2+\Delta_d$, $p_{\mathbf{T}}(t^1,t^2) = $
%\begin{eqnarray}
%P_s(1)\int_0^\infty p_{\mathbf{T}}(t^1+x-\Delta_s, t^2+x)p_xdx
%\end{eqnarray}
%When $t^2\in(-\Delta_d,0)$, $t_1=t_2+\Delta_d$, $p_{\mathbf{T}}(t^1,t^2) = $
%\begin{eqnarray}
%P_s(1)\int_0^\infty \int_{-\Delta_d}^{t^2+x-\Delta_d}p_{\mathbf{T}}(\tau, t^2+x)d\tau p_xdx
%\end{eqnarray}
%When $t^2>0$, $t_1>t_2+\Delta_d$, $p_{\mathbf{T}}(t^1,t^2) =$
%\begin{eqnarray}
%& P_s(1)\int_0^\infty p_{\mathbf{T}}(t^1+x-\Delta_s, t^2+x)p_xdx \nonumber\\
%& + P_s(2) \int_0^\infty p_{\mathbf{T}}(t^1+x, t^2+x-\Delta_s)p_xdx
%\end{eqnarray}
%When $t^2>0$, $t_1=t_2+\Delta_d$, $p_{\mathbf{T}}(t^1,t^2) =$
%\begin{eqnarray}
% P_s(2) \int_0^\infty \int_{\Delta_d-\Delta_s}^{\Delta_d} p_{\mathbf{T}}(t^2+x+\tau, t^2+x-\Delta_s)d\tau p_xdx
%\end{eqnarray}
%When $t^2=0$, $t^1=\Delta$, $p_{\mathbf{T}}(t^1,t^2) = $
%\begin{eqnarray}
%P_s(2) \int_0^\infty \int_{-\Delta_d}^{x-\Delta_s} \int_x^{x+\Delta_d} p_{\mathbf{T}}(\tau_1,\tau_2)d\tau_1d\tau_2 p_xdx
%\end{eqnarray}
%When $t^2=0$, $t^1>\Delta$, $p_{\mathbf{T}}(t^1,t^2) = $
%\begin{eqnarray}
%P_s(2) \int_0^\infty \int_{-\Delta_d}^{x-\Delta_s} p_{\mathbf{T}}(t^1+x,\tau)d\tau p_xdx
%\end{eqnarray}
The distribution obtained from EDS with the same condition as in the FIFO case is shown in \Cref{fig: fo zb}. FO generates smaller delay as compared to FIFO, but FO no longer has the ``zebra'' pattern shown in FIFO. %In the following discussion, define $y= e^{-\lambda\Delta_d}$ and $y_i = e^{-\lambda_i\Delta_d}$ for $i\in\{1,2\}$. 

In the following discussion, we discuss the necessary condition for convergence under FO and derive the exact steady state distribution of delay for $\Delta_s=0$. The distribution for $\Delta_s>0$ is left as future work. %\Cref{cor: scalar delay FO} derives the distribution of vehicle delay. 
Recall the definitions $y= e^{-\lambda\Delta_d}$ and $y_i = e^{-\lambda_i\Delta_d}$ for $i\in\{1,2\}$.

\begin{prop}[Necessary Condition for Convergence under FO]\label{prop: convergence FO}
The distributions $\{p_{\mathbf{T}_i}\}_i$ converges for FO only if the following condition holds
\begin{equation}
\lambda_1\lambda_2(y_1+y_2)\Delta_d+[\lambda_1^2+\lambda_2^2+\lambda_1\lambda_2(2-y_1-y_2)]\Delta_s \leq \lambda\text{.}\label{eq: convergence FO}
\end{equation}
\end{prop}
\begin{proof}
Similar to the discussion in \Cref{prop: convergence FIFO}, the minimum average departure interval between two consecutive vehicles should be smaller than the average arrival interval. As FO adjusts the passing order, vehicles from the same lane may be grouped and pass the intersection together. For two vehicles consecutively leaving the intersection, they go to different lanes only if the following two conditions holds: 1) they come from different lanes and 2) the last vehicle in the two has a temporal gap greater than $\Delta_d$ with its front vehicle. Hence, the probability that two departure vehicles are from different lanes  is $P_s(1)P_s(2)(e^{-\lambda_1\Delta_d}+e^{-\lambda_2\Delta_d})$, which is smaller than $2P_s(1)P_s(2)$. The minimum average departure interval is $P_s(1)P_s(2)(y_1+y_2)\Delta_d+\left[P_s(1)^2+P_s(2)^2+P_s(1)P_s(2)(2-y_1-y_2)\right]\Delta_s$. The average arrival interval is $\frac{1}{\lambda}$. Condition \eqref{eq: convergence FIFO} can be obtained by requiring the minimum departure interval be smaller than the arrival interval.
\end{proof}

%In this paper, we investigate the steady state distribution for $\Delta_s = 0$ and leave the case that $\Delta_s > 0$ for future work. Indeed, when $\Delta_s = 0$, the mapping $p_{\mathbf{T}_i}\mapsto p_{\mathbf{T}_{i+1}}$ is a contraction as illustrated in \Cref{fig: FO convergence} for any random initial distribution. In  \Cref{prop: fo delta_s=0}, we derive a closed form solution of $p_{\mathbf{T}}$ for $\Delta_s = 0$.

%Define $g_1(t) := p_{\mathbf{T}}(t,t-\Delta_d)$, $g_2(t) := p_{\mathbf{T}}(t-\Delta_d,t)$. %Both functions have point mass at $t=0$ and $t=\Delta_d$. 
%Denote finite part of the functions as $\tilde g_1(t)$ and $\tilde g_2(t)$, and the delta components as as $\widehat g_1(t)$ and $\widehat g_2(t)$. Moreover, define $\mathcal{M}_1 = \int_0^{\infty} g_1(x)dx,
%\mathcal{M}_2 = \int_0^{\infty} g_2(x)dx$ and $\mathcal{I}_1 = \int_0^{\infty} g_1(x)e^{-\lambda x}dx,
%\mathcal{I}_2 = \int_0^{\infty} g_2(x)e^{-\lambda x}dx$.

\begin{figure}[t]
\begin{center}
\vspace{-10pt}
\subfloat[Domain]{
\includegraphics[width=4.3cm]{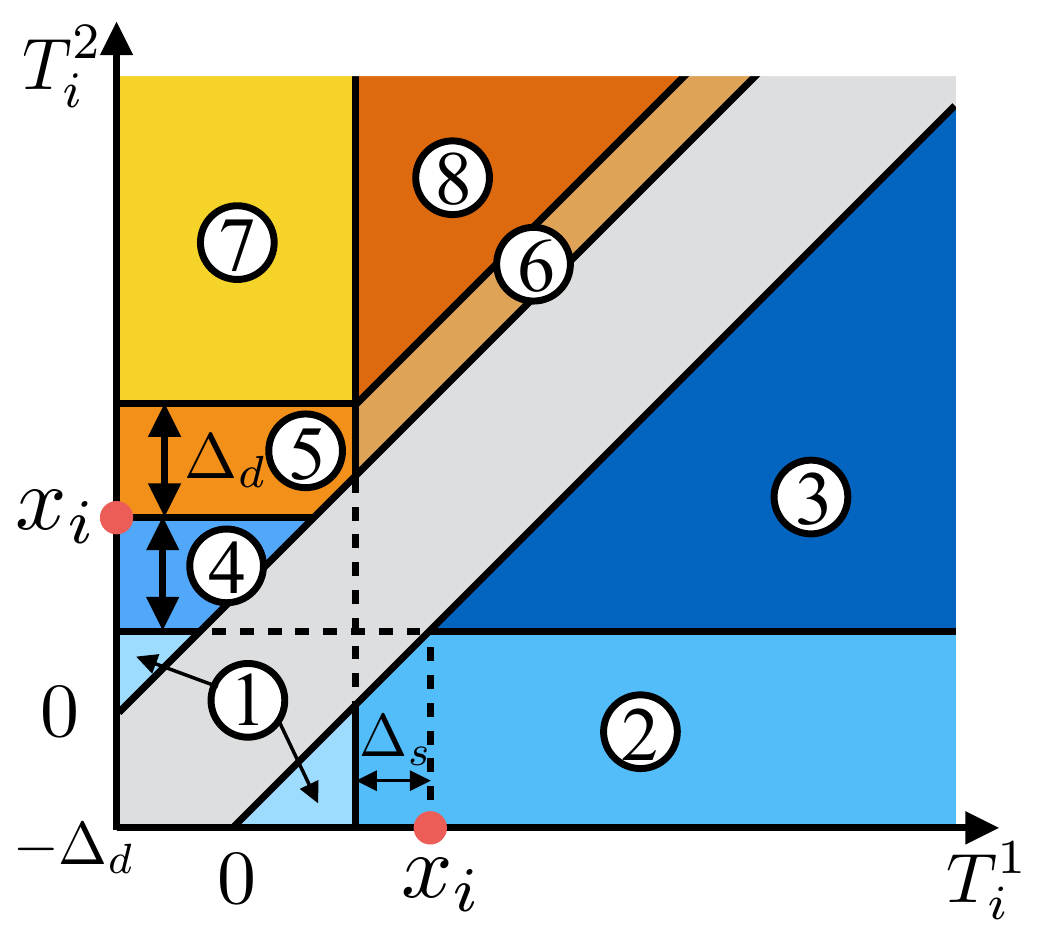}}
\subfloat[Value]{
\includegraphics[width=4.1cm]{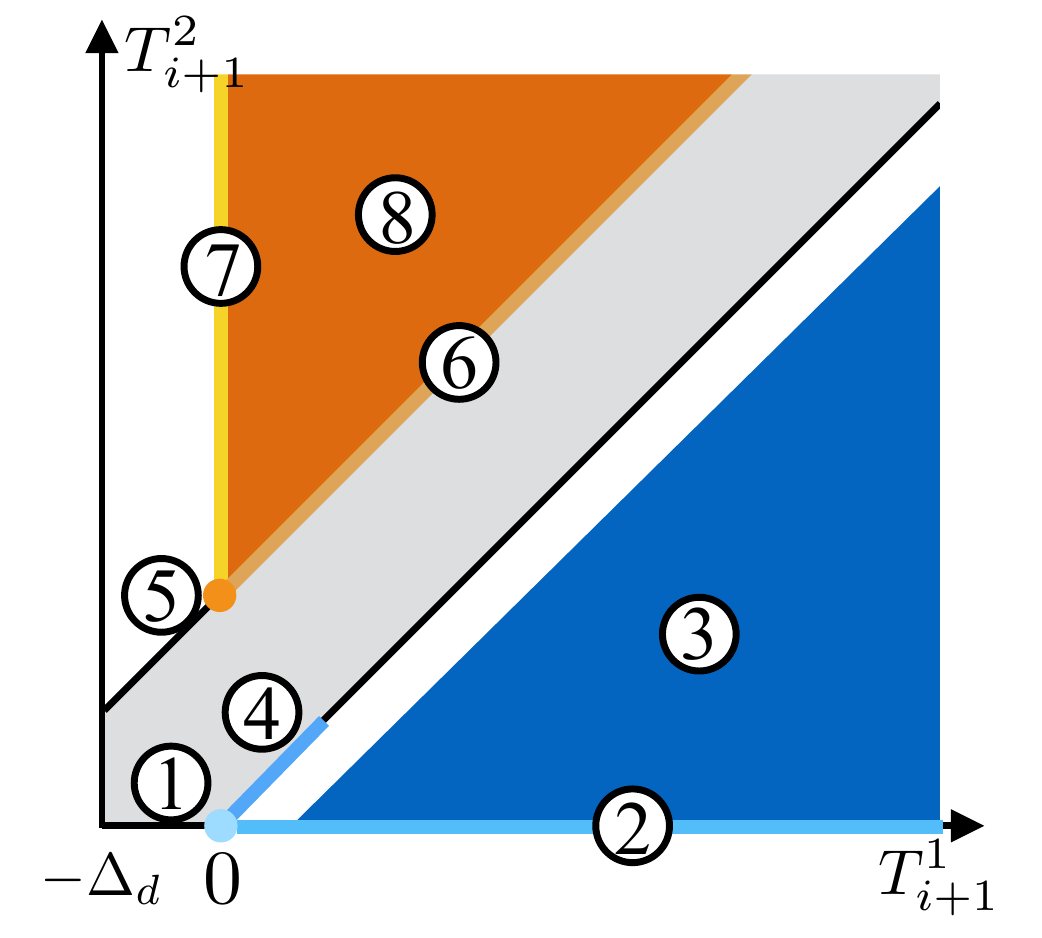}}
\caption{Illustration of the mapping \eqref{eq: dynamic} under FO for $s_{i+1} = 1$.}
\label{fig: FO mapping}
\end{center}
%\vspace{-20pt}
\end{figure}

\begin{table}[t]
\caption{The mapping \eqref{eq: dynamic} under FO for $s_{i+1} = 1$.}
\vspace{-10pt}
\begin{center}
\begin{tabular}{ccc}
\toprule
 & Domain & Value\\
\midrule
1 & $\begin{array}{c}T_i^1<x_i-\Delta_s \\ T_i^2<x_i-\Delta_d\end{array}$ & $\begin{array}{cc} T_{i+1}^1 = 0 \\ T_{i+1}^2 = -\Delta_d\end{array}$\\
\midrule
2 & $\begin{array}{c}T_i^1\geq x_i-\Delta_s \\ T_i^2<x_i-\Delta_d \\ T_i^2<T_i^1\end{array}$ & $\begin{array}{cc} T_{i+1}^1 = T_i^1+\Delta_s-x_i \\ T_{i+1}^2 = -\Delta_d\end{array}$\\
\midrule
3 & $\begin{array}{c}T_i^2\geq x_i-\Delta_d \\ T_i^2<T_i^1\end{array}$ & $\begin{array}{cc} T_{i+1}^1 = T_i^1+\Delta_s-x_i \\ T_{i+1}^2 = T_i^2-x_i\end{array}$\\
\midrule
4 & $\begin{array}{c}T_i^2\in [x_i-\Delta_d, x_i) \\ T_i^2>T_i^1\end{array}$ & $\begin{array}{cc} T_{i+1}^1 = T_i^2+\Delta_d-x_i \\ T_{i+1}^2 = T_i^2-x_i\end{array}$\\
\midrule
5 & $\begin{array}{c}T_i^2\in [x_i, x_i+\Delta_d) \\ T_i^1<x_i-\Delta_s \end{array}$ & $\begin{array}{cc} T_{i+1}^1 = 0 \\ T_{i+1}^2 = \Delta_d \end{array}$\\
\midrule
6 & $\tiny\begin{array}{c}T_i^2-T_i^1\in [\Delta_d,\Delta_d+\Delta_s] \\ T_i^1\geq x_i-\Delta_s \end{array}$ & $\tiny\begin{array}{cc} T_{i+1}^1 = T_i^1-x_i+\Delta_s \\ T_{i+1}^2 = T_i^1-x_i+\Delta_s+\Delta_d \end{array}$\\
\midrule
7 & $\begin{array}{c}T_i^1<x_i-\Delta_s \\ T_i^2\geq x_i+\Delta_d \end{array}$ & $\begin{array}{cc} T_{i+1}^1 = 0 \\ T_{i+1}^2 = T_i^2-x_i \end{array}$\\
\midrule
8 & $\begin{array}{c}T_i^2-T_i^1 > x_i+\Delta_d+\Delta_s \\ T_i^1\geq x_i-\Delta_s \end{array}$ & $\begin{array}{cc} T_{i+1}^1 = T_i^1-x_i+\Delta_s \\ T_{i+1}^2 = T_i^2-x_i \end{array}$\\
\bottomrule
\end{tabular}
\end{center}
\label{table: FO mapping}
\end{table}%

\begin{prop}[Steady State Distribution for $\Delta_s=0$ under FO]\label{prop: fo delta_s=0}
If $\Delta_s = 0$, then $p_{\mathbf{T}}(t^1,t^2)=0$ if $|t^1-t^2|\neq -\Delta_d$ or $t^1+t^2>\Delta_d$. Moreover, for $i\in\{1,2\}$,
\begin{eqnarray}
G_i(t) &=& \left\{\begin{array}{ll} \frac{c_i}{\lambda_{i^*}} e^{\lambda_{i^*} t} & t\in[0,\Delta_d)\\
\mathcal{M}_i & t\geq\Delta_d
\end{array}\right.\text{,}\label{eq: density FO}\\
%\widehat g_i(t) &=& \left\{\begin{array}{ll}\frac{\lambda_i}{\lambda}[\mathcal{I}_i+e^{-\lambda\Delta_d}\mathcal{I}_{i^*}] & t=0\\\frac{\lambda_i}{\lambda}[\mathcal{M}_{i^*}-\mathcal{I}_{i^*}] & t=\Delta_d\\0 & \text{otherwise}\end{array}\right.\text{,}\label{eq: point mass FO}\\
\mathcal{M}_i &=& \frac{\lambda_i}{\lambda}\text{,}\label{eq: integration M FO}\\
%\mathcal{I}_i &=& \frac{\lambda_i}{\lambda}e^{-\lambda\Delta_d} + \frac{\lambda^3 c_i(1-e^{-\lambda_i\Delta_d})}{\lambda_i\lambda_{i^*}}\text{,}\label{eq: integration I FO}\\
c_i &=& \frac{\lambda_{i}\lambda_{i^*} \left(\lambda_{i} y^2   + \lambda_{i}   y_{i^*}  + \lambda_{i^*}  y- \lambda_{i}  y^2y_{i^*} \right)}{{\lambda}^2  \left( 1+yy_i + yy_{i^*}-  y-y^2 \right)}\text{.}\label{eq: constant FO}%\frac{\lambda_i\lambda_{i^*}[\lambda_{i^*} +\lambda_ie^{\lambda_i\Delta_d}+\lambda_ie^{-\lambda\Delta_d}(1-e^{-\lambda_{i^*}\Delta_d})]}{\lambda^2[e^{\lambda\Delta_d}-e^{-\lambda\Delta_d}+e^{-\lambda_i\Delta_d}+e^{-\lambda_{i^*}\Delta_d}-1]}\text{.}\label{eq: constant FO}
\end{eqnarray}
\end{prop}
\begin{proof}
Similar to the proof of  \Cref{prop: fifo delta_s=0}, it is easy to show that $p_{\mathbf{T}}(t^1,t^2)=0$ if $|t^1-t^2|\neq -\Delta_d$. For $t>\Delta_d$, consider regions 3 and 6, $g_i(t) = \int_0^\infty \left[P_s(i)g_i(t+x)+P_s(i^*)g_{i}(t+x)\right]p_xdx$. Hence,
\begin{equation}
g_i(t) = \lambda \int_0^\infty g_i(t+x)e^{-\lambda x}dx\text{.}
\end{equation}
According to \Cref{lem: zero function}, $g_i(t)\equiv 0$ for $t>\Delta_d$. Hence, $p_{\mathbf{T}}(t^1,t^2)=0$ if $|t^1-t^2|\neq -\Delta_d$ or $t^1+t^2>\Delta_d$.

For $t\in(0,\Delta_d)$, consider regions 3 and 4, \eqref{eq: g1 formula} holds. Similar to the proof in \Cref{prop: fifo delta_s=0} from \eqref{eq: g1 formula} to \eqref{eq: c FIFO}, we conclude that $\tilde g_i = c_i e^{\lambda_{i^*} t}$ for some constant $c_i$ such that 
\begin{eqnarray}
c_i = \lambda_{i^*}\widehat g_i(0)\text{.}\label{eq: c FO}
\end{eqnarray}
Then \eqref{eq: density FO} is verified. We solve for $c_i$ below.
%\begin{eqnarray*}
%\int_{0}^{\Delta_d} \tilde g_1dt = \lambda_1 \mathcal{I}_2 \int_{0}^{\Delta_d} e^{\lambda(t-\Delta_d)}dt + \lambda_1 \int_{0}^{\Delta_d} e^{\lambda t}\int_t^{\infty} g_1(x)e^{-\lambda x}dx dt
%\end{eqnarray*}\normalsize
%Note that $\int_{0}^{\Delta_d} \tilde g_1dt = \frac{C_1(e^{\lambda_2\Delta_d}-1)}{\lambda_2}$, $ \lambda_1 \mathcal{I}_2 \int_{0}^{\Delta_d} e^{\lambda(t-\Delta_d)}dt= P_s(1)\mathcal{I}_2(1-e^{-\lambda\Delta_d})$ and
%\begin{eqnarray*}
%&&\lambda_1 \int_{0}^{\Delta_d} e^{\lambda t}\int_t^{\infty} g_1(x)e^{-\lambda x}dxdt\\
%&=& \lambda_1\int_0^{\Delta_d}\int_0^{x}e^{\lambda t} dt  g_1(x)e^{-\lambda x}dx\\
%&=& P_s(1)\int_0^{\Delta_d}(e^{\lambda x}-1)  g_1(x)e^{-\lambda x}dx\\
%&=& P_s(1)\left[\mathcal{M}_1-\mathcal{I}_1\right]
%\end{eqnarray*}
%Hence
%\begin{eqnarray}
%C_1 = \frac{P_s(1)\lambda_2}{e^{\lambda_2\Delta_d}-1}\left[\mathcal{I}_2(1-e^{-\lambda\Delta_d})+\mathcal{M}_1-\mathcal{I}_1\right]
%\end{eqnarray}

Consider region 1. The point mass at $0$ has the same expression as in the FIFO case, 
\begin{equation}
\widehat g_i(0) = \frac{\lambda_i}{\lambda}\left[\mathcal{I}_i+e^{-\lambda\Delta_d}\mathcal{I}_{i^*}\right]\text{.}
\end{equation}

Consider region 5. The point mass at $\Delta_d$ satisfies $\widehat g_i(\Delta_d) =\lambda_i\int_0^{\infty}\int_x^{\Delta_d} g_{i^*}(\tau)d\tau e^{-\lambda x}dx$. By changing the order of integration, we have $\widehat g_i(\Delta_d)= \lambda_i\int_0^{\Delta_d}\int_0^{\tau} e^{-\lambda x}dxg_{i^*}(\tau)d\tau= \frac{\lambda_i}{\lambda}\int_0^{\Delta_d}(1-e^{-\lambda\tau})g_{i^*}(\tau)d\tau$. Hence,
\begin{equation}
\widehat g_i(\Delta_d) = \frac{\lambda_i}{\lambda}[\mathcal{M}_{i^*}-\mathcal{I}_{i^*}]\text{.}
\end{equation}

%\begin{eqnarray}
%\frac{\lambda_{1}  \lambda_{2}  \left(\lambda_{1}   e^{\Delta_{d}  \left(\lambda + \lambda_{1}\right)} + \lambda_{1}   e^{2  \Delta_{d}  \left(\lambda + \lambda_{1}\right)} - \lambda_{1}   e^{2  \Delta_{d}  \lambda_{1}} + \lambda   e^{\Delta_{d}  \left(2  \lambda + \lambda_{1}\right)} - \lambda_{1}   e^{\Delta_{d}  \left(2  \lambda + \lambda_{1}\right)}\right)}{{\lambda}^2  \left( e^{2  \Delta_{d}  \lambda} -  e^{\Delta_{d}  \left(\lambda + \lambda_{1}\right)} +  e^{\Delta_{d}  \left(\lambda + 2  \lambda_{1}\right)} -  e^{\Delta_{d}  \left(2  \lambda + \lambda_{1}\right)} +  e^{\Delta_{d}  \left(3  \lambda + \lambda_{1}\right)}\right)}
%\end{eqnarray}

Given the definition in \eqref{eq: defn G M I}, 
\begin{eqnarray}
\mathcal{M}_i &=& \widehat g_i(0) + \frac{c_i}{\lambda_{i^*}}\left[e^{\lambda_{i^*}\Delta_d}-1\right] + \widehat g_i(\Delta_d)\text{,}\\
\mathcal{I}_i &=& \widehat g_i(0)+ \frac{c_i}{\lambda_i}\left[1-e^{-\lambda_i\Delta_d}\right]+\widehat g_i(\Delta_d)e^{-\lambda\Delta_d}\text{.}
\end{eqnarray}

Moreover, the probability should add up to one,
\begin{equation}
\mathcal{M}_1 + \mathcal{M}_2 = 1\text{.}\label{eq: integral to one}
\end{equation}
Solving \eqref{eq: c FO} to \eqref{eq: integral to one}, we conclude that $\mathcal{M}_i=\frac{\lambda_i}{\lambda}$ and $c_i$ satisfies \eqref{eq: constant FO}.

\end{proof}

%The result that the two lanes have the same traffic density $\lambda_1=\lambda_2=\frac{\lambda}{2}$ was shown in \cite{liu2018analytically}. 

%\Cref{prop: fo delta_s=0} shows that we can obtain exact distribution under FO, while  for \Cref{prop: fo delta_s=0} derives the exact steady state distribution of delay for $\Delta_s=0$. The distribution for $\Delta_s>0$ is left as future work. \Cref{cor: scalar delay FO} derives the distribution of vehicle delay. Define $y= e^{-\lambda\Delta_d}$ and $y_i = e^{-\lambda_i\Delta_d}$ for $i\in\{1,2\}$.

\begin{cor}[Steady State Vehicle Delay under FO]\label{cor: scalar delay FO}
The steady state vehicle delay under FO has the distribution
\begin{eqnarray}
P_d(t) &=& \frac{c_2}{\lambda_1} e^{\lambda_1 t} + \frac{c_1}{\lambda_2} e^{\lambda_2 t} + \frac{2\lambda_1\lambda_2}{\lambda^2}(1- e^{-\lambda t}) \label{eq: Pd FO}\\
&&+ \frac{c_2}{\lambda_2y_1}(e^{-\lambda t}- e^{-\lambda_1 t}) + \frac{c_1}{\lambda_1y_2}(e^{-\lambda t} -e^{-\lambda_2 t})\text{,}\nonumber
\end{eqnarray}
with expected delay
\begin{eqnarray}
E(d) &=& \frac{c_2}{\lambda_1} \mathcal{E}(\lambda_1) + \frac{c_1}{\lambda_2} \mathcal{E}(\lambda_2) \nonumber\\
&& - \frac{c_2}{\lambda_2y_1}\mathcal{E}(-\lambda_1) - \frac{c_1}{\lambda_1y_2}\mathcal{E}(-\lambda_2)\nonumber\\
&& +  \left(\frac{c_2}{\lambda_2y_1} + \frac{c_1}{\lambda_1y_2} -\frac{2\lambda_1\lambda_2}{\lambda^2}\right)\mathcal{E}(-\lambda)\text{,}\label{eq: Ed FO}
\end{eqnarray}
where $\mathcal{E}(\cdot)$ follows \eqref{eq: integral 0 delta_d}.
\end{cor}
\begin{proof}
By \eqref{eq: scalar delay and traffic delay}, the steady state distribution of delay satisfies $p_d(t) = \sum_{i = 1,2}P_s(i)\int_0^{\infty} [g_i(t+x)+g_{i^*}(t+x-\Delta_d)+g_{i^*}(x-t+\Delta_d)]p_{x}dx$. Using the result from \Cref{prop: fo delta_s=0}, the steady state distribution of the vehicle delay satisfies \eqref{eq: Pd FO}. 
%\begin{eqnarray}
%P_d(t) = \frac{c_1}{\lambda_2}(e^{\lambda_2t}-y_1e^{\lambda t})+\frac{c_2}{\lambda_1}(e^{\lambda_1t}-y_2e^{\lambda t})\\
%+2\frac{\lambda_1\lambda_2}{\lambda^2}+\frac{\lambda_1^2+\lambda_2^2}{\lambda^2}ye^{\lambda t}+\frac{\lambda_1\lambda_2}{\lambda^2}\left[y^2e^{\lambda t}-e^{-\lambda t}\right]\\
%+e^{\lambda t} \left[\frac{c_1}{\lambda_1}(y-yy_1)+\frac{c_2}{\lambda_2}(y-yy_2)\right]\\
%-\frac{c_2}{\lambda_2}(\frac{e^{-\lambda_1 t}-e^{-\lambda t}}{y_1})-\frac{c_1}{\lambda_1}(\frac{e^{-\lambda_2 t}-e^{-\lambda t}}{y_2})
%\end{eqnarray}
%\begin{subequations}
%\begin{align}
%P_d(t) &= S_1 + S_2 e^{\lambda_1 t} + S_3 e^{\lambda_2 t} + S_4 e^{-\lambda t} + S_5 e^{-\lambda_1 t} + S_6 e^{-\lambda_2 t} \\
%S_1 &= 2\frac{\lambda_1\lambda_2}{\lambda^2}\\
%%S_1 &= \frac{c_1}{\lambda_1}(y-yy_1)+\frac{c_2}{\lambda_2}(y-yy_2)-\frac{c_1y_1}{\lambda_2}-\frac{c_2y_2}{\lambda_1}\nonumber\\\&+\frac{\lambda_1^2+\lambda_2^2}{\lambda^2}y+\frac{2\lambda_1\lambda_2}{\lambda^2}y^2\\%-\frac{y  \left(\lambda_2 -  \lambda_{1}\right)  \left((\lambda_{1}  y - \lambda_2  {y_1}^2)(y+1) + (\lambda_2-\lambda_1)  y  y_1 \right)}{{\lambda}^2  \left( - y^2  y_1 + y^2 + y  {y_1}^2 - y  y_1 + y_1\right)}\\
%S_2 &= \frac{c_2}{\lambda_1}\\
%S_3 &= \frac{c_1}{\lambda_2}\\
%S_4 &= -\frac{2\lambda_1\lambda_2}{\lambda^2}+\frac{c_2}{\lambda_2y_1}+\frac{c_1}{\lambda_1y_2}\\
%S_5 &= -\frac{c_2}{\lambda_2y_1}\\
%S_6 &= -\frac{c_1}{\lambda_1y_2}
%\end{align}
%\end{subequations}
It is easy to verify $P_d(0) = \frac{c_2}{\lambda_1}  + \frac{c_1}{\lambda_2}  = \widehat g_1(0)+\widehat g_2(0)$ and $P_d(\Delta_d) = 1$.
The expected mean $E(d) = \int_0^{\Delta_d} t dP_d(t)$ satisfies \eqref{eq: Ed FO}.
\end{proof}

\Cref{cor: scalar delay FO} implies that the distribution of vehicle delay in FO no longer equals the sum of traffic delay in all lanes. In FIFO, the two equal by \Cref{cor: scalar delay FIFO}.

\section{Analysis\label{sec: result}}
This section discusses how delay is affected by traffic density $\lambda$, density ratio $r:=\lambda_1/\lambda_2$, passing order (FIFO or FO), and temporal gap $\Delta_d$. $\Delta_s=0$ is assumed. In particular, we evaluate the probability of zero delay $P_d(0)=\widehat g_1(0)+\widehat g_2(0)$ in \cref{fig: zero delay}, expected delay $E(d)$ in \cref{fig: expected delay}, and steady state distribution of delay $P_d(t)$ in \cref{fig: distribution}. The curves are from direct analysis. Approximation \eqref{eq: FIFO approx 1} is used for FIFO. The accuracy of the analytical solutions is verified by EDS in \Cref{fig: expected delay delta}.

\begin{figure}[t]
%\begin{center}
\subfloat[Fix $\Delta_d = \SI{2}{\second}$.\label{fig: zero delay lambda}]{
\input{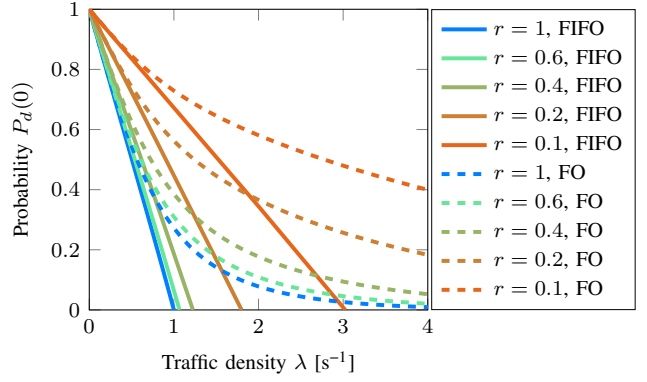}}\\
\subfloat[Fix $\lambda = \SI{1}{\per\second}$.\label{fig: zero delay ratio}]{
\input{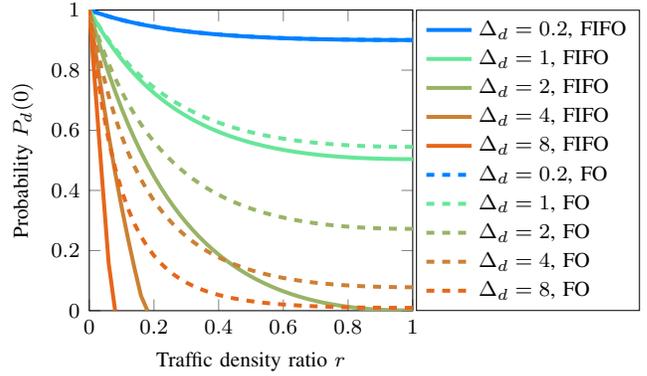}}\\
\subfloat[Fix $r = 0.5$.\label{fig: zero delay delta}]{
\input{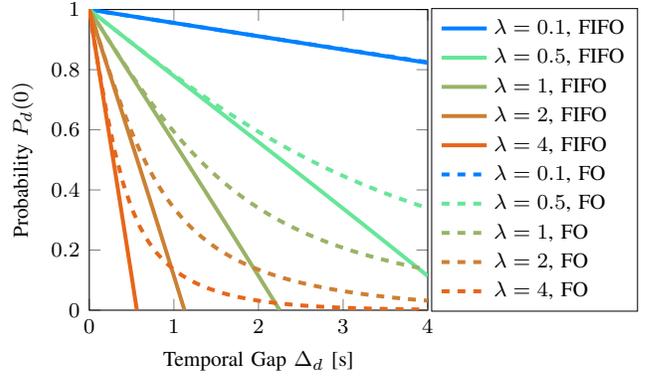}}
\caption{The probability of zero delay $P_d(0)$.}\label{fig: zero delay}
%\end{center}
\end{figure}

\begin{figure}[t]
%\begin{center}
\subfloat[Fix $\Delta_d = \SI{2}{\second}$.\label{fig: expected delay lambda}]{
\input{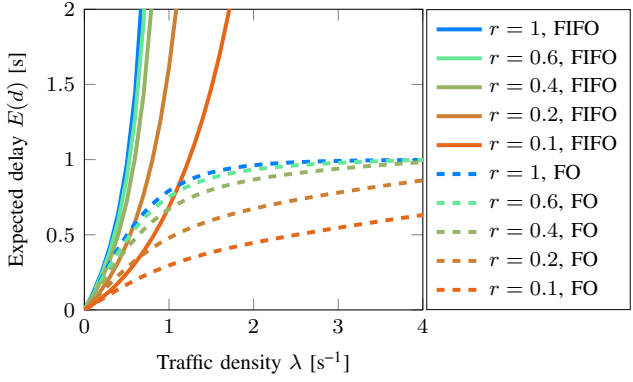}}\\
\subfloat[Fix $\lambda = \SI{1}{\per\second}$.\label{fig: expected delay ratio}]{
\input{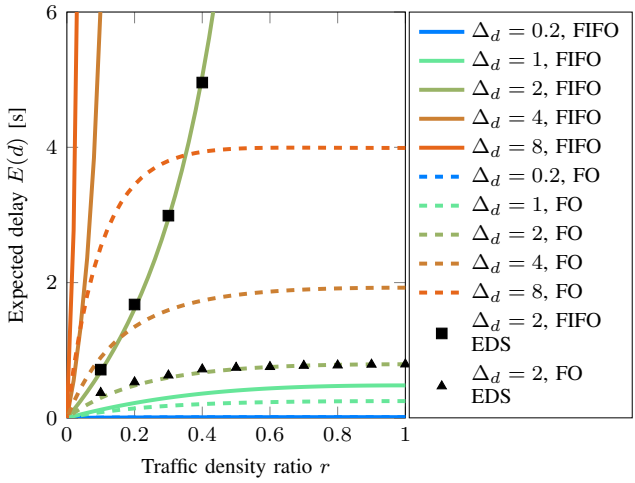}
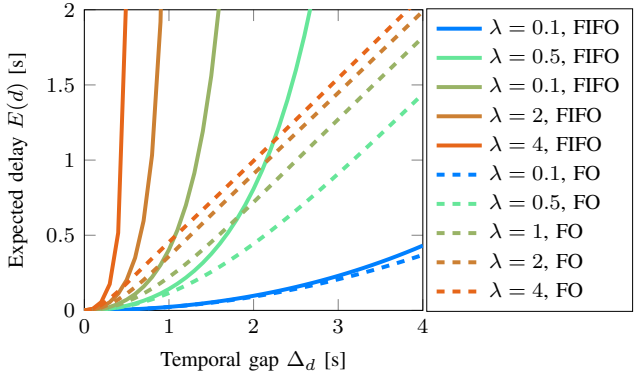}\\
\subfloat[Fix $r = 0.5$.\label{fig: expected delay delta}]{
\begin{tikzpicture}

\begin{axis}[%
width=4.5cm,
height=4cm,
at={(1.011in,0.642in)},
scale only axis,
xmin=0,
xmax=4,
ymin=0,
ymax=2,
xlabel={Temporal gap $\Delta_d$ [\si{\second}]},
ylabel={Expected delay $E(d)$ [\si{\second}]},
font=\footnotesize,
axis background/.style={fill=white},
legend style={legend cell align=left, align=left, draw=white!6!black, font=\footnotesize},
%legend pos=north east,
legend style={at={(1.01,0.5)},anchor=west}
]
\addplot [color=mycolor1, line width=1.5pt]
  table[row sep=crcr]{%
  0	0\\
0.1	0.000223112876480007\\
0.2	0.000896157190417115\\
0.3	0.00202487862314924\\
0.4	0.00361520954293046\\
0.5	0.00567324066887613\\
0.6	0.00820520985217327\\
0.7	0.0112174970451992\\
0.8	0.014716622279\\
0.9	0.018709245475804\\
1	0.0232021673207697\\
1.1	0.0282023309300569\\
1.2	0.0337168240731156\\
1.3	0.0397528818218567\\
1.4	0.0463178895401516\\
1.5	0.0534193861552468\\
1.6	0.0610650676708266\\
1.7	0.0692627916541994\\
1.8	0.0780205784842292\\
1.9	0.0873466190550655\\
2	0.0972492768648466\\
2.1	0.107737092954449\\
2.2	0.118818790505383\\
2.3	0.130503279596323\\
2.4	0.142799662117755\\
2.5	0.15571723684544\\
2.6	0.169265504674337\\
2.7	0.18345417401547\\
2.8	0.19829316635888\\
2.9	0.213792632824759\\
3	0.229962920160189\\
3.1	0.246814632560442\\
3.2	0.264358603142316\\
3.3	0.282605908601693\\
3.4	0.301567875877049\\
3.5	0.321256089016364\\
3.6	0.341682396254351\\
3.7	0.362858917307299\\
3.8	0.384798050893223\\
3.9	0.407512482485361\\
4	0.431015192307477\\
};
\addlegendentry{$\lambda = 0.1$, FIFO}

\addplot [color=mycolor2, line width=1.5pt]
  table[row sep=crcr]{%
  0	0\\
0.1	0.00113464813415671\\
0.2	0.00464043346916943\\
0.3	0.0106838773029284\\
0.4	0.0194498555113157\\
0.5	0.0311434481156397\\
0.6	0.0459925849802807\\
0.7	0.0642512213291901\\
0.8	0.0862030527022651\\
0.9	0.112165825281979\\
1	0.14249638531198\\
1.1	0.17759654544169\\
1.2	0.217920072309394\\
1.3	0.263980805039392\\
1.4	0.316362570524537\\
1.5	0.375730619879822\\
1.6	0.442846114266206\\
1.7	0.518582715629159\\
1.8	0.603947458716659\\
1.9	0.700107175163099\\
2	0.808418882702805\\
2.1	0.93047021726689\\
2.2	1.06812910220748\\
2.3	1.22360584365843\\
2.4	1.39953952915285\\
2.5	1.59909967398209\\
2.6	1.82612065624312\\
2.7	2.08529894224421\\
2.8	2.38243629891611\\
2.9	2.72479622248859\\
3	3.12156580000804\\
3.1	3.58456483980431\\
3.2	4.12924888926378\\
3.3	4.77622096166126\\
3.4	5.55360764801445\\
3.5	6.50076035972688\\
3.6	7.67444788792846\\
3.7	9.15981986061028\\
3.8	11.0908790255363\\
3.9	13.6907039166141\\
4	17.3604961456793\\
};
\addlegendentry{$\lambda = 0.5$, FIFO}

\addplot [color=mycolor3, line width=1.5pt]
  table[row sep=crcr]{%
  0	0\\
0.1	0.00232021673473125\\
0.2	0.00972492775928612\\
0.3	0.0229962925142535\\
0.4	0.043101526448842\\
0.5	0.071248192968207\\
0.6	0.108960037035903\\
0.7	0.158181302002907\\
0.8	0.221423103510468\\
0.9	0.301973855799483\\
1	0.404209626403348\\
1.1	0.534064859773907\\
1.2	0.699771266042273\\
1.3	0.913061600776453\\
1.4	1.19122320773131\\
1.5	1.560790140904\\
1.6	2.06463883936337\\
1.7	2.77684446416148\\
1.8	3.83729613227624\\
1.9	5.54570486527913\\
2	8.68101899601102\\
2.1	16.098179264254\\
};
\addlegendentry{$\lambda = 0.1$, FIFO}

\addplot [color=mycolor4, line width=1.5pt]
  table[row sep=crcr]{%
  0	0\\
0.1	0.00486246388178127\\
0.2	0.0215507634510329\\
0.3	0.054480021191979\\
0.4	0.110711562079086\\
0.5	0.202104860856488\\
0.6	0.349885804994925\\
0.7	0.595612711040567\\
0.8	1.03232330672864\\
0.9	1.91866740806624\\
1	4.34070494612481\\
1.1	26.7277886322575\\
};
\addlegendentry{$\lambda = 2$, FIFO}

\addplot [color=mycolor5, line width=1.5pt]
  table[row sep=crcr]{%
  0	0\\
0.1	0.0107753817374846\\
0.2	0.0553557815321253\\
0.3	0.17494295265822\\
0.4	0.516163003556102\\
0.5	2.17040729496102\\
};
\addlegendentry{$\lambda = 4$, FIFO}

\addplot [color=mycolor1, dashed, line width=1.5pt]
  table[row sep=crcr]{%
  0	0\\
0.1	0.000222588778876144\\
0.2	0.000891790720521536\\
0.3	0.00200968996783389\\
0.4	0.00357827768791304\\
0.5	0.00559945133756584\\
0.6	0.00807501398296904\\
0.7	0.0110066736746728\\
0.8	0.0143960428790008\\
0.9	0.0182446379667788\\
1	0.0225538787602057\\
1.1	0.0273250881385611\\
1.2	0.0325594917033321\\
1.3	0.0382582175032295\\
1.4	0.0444222958194544\\
1.5	0.0510526590114728\\
1.6	0.0581501414234507\\
1.7	0.0657154793514013\\
1.8	0.073749311071002\\
1.9	0.0822521769259376\\
2	0.0912245194765404\\
2.1	0.100666683708403\\
2.2	0.110578917300553\\
2.3	0.120961370952704\\
2.4	0.13181409877099\\
2.5	0.143137058711556\\
2.6	0.154930113081245\\
2.7	0.167193029094604\\
2.8	0.179925479486322\\
2.9	0.193127043178161\\
3	0.206797205999381\\
3.1	0.220935361459592\\
3.2	0.235540811572904\\
3.3	0.250612767732206\\
3.4	0.266150351632329\\
3.5	0.282152596240833\\
3.6	0.298618446815073\\
3.7	0.315546761964188\\
3.8	0.332936314754598\\
3.9	0.350785793857575\\
4	0.369093804737391\\
};
\addlegendentry{$\lambda = 0.1$, FO}

\addplot [color=mycolor2, dashed, line width=1.5pt]
  table[row sep=crcr]{%
  0	0\\
0.1	0.00111989026751317\\
0.2	0.00451077575204114\\
0.3	0.0102105318022946\\
0.4	0.0182449038953081\\
0.5	0.0286274117423112\\
0.6	0.0413594411998762\\
0.7	0.0564305192481666\\
0.8	0.0738187609474783\\
0.9	0.0934914720555434\\
1	0.115405886931042\\
1.1	0.139510018486029\\
1.2	0.165743595262164\\
1.3	0.194039060127537\\
1.4	0.224322605514844\\
1.5	0.256515221406187\\
1.6	0.290533734250113\\
1.7	0.326291817497052\\
1.8	0.363700957284322\\
1.9	0.402671359825878\\
2	0.443112790116856\\
2.1	0.484935334523509\\
2.2	0.528050082595858\\
2.3	0.572369725939792\\
2.4	0.617809074169576\\
2.5	0.664285489805237\\
2.6	0.711719245476356\\
2.7	0.760033807954318\\
2.8	0.809156054380895\\
2.9	0.859016426622492\\
3	0.909549029991681\\
3.1	0.960691682678525\\
3.2	1.01238592216134\\
3.3	1.06457697465618\\
3.4	1.1172136933496\\
3.5	1.17024847077065\\
3.6	1.22363713022009\\
3.7	1.27733880071114\\
3.8	1.33131577940302\\
3.9	1.38553338504056\\
4	1.4399598054625\\
};
\addlegendentry{$\lambda = 0.5$, FO}

\addplot [color=mycolor3, dashed, line width=1.5pt]
  table[row sep=crcr]{%
  0	0\\
0.1	0.00225538787602057\\
0.2	0.00912245194765404\\
0.3	0.0206797205999381\\
0.4	0.0369093804737392\\
0.5	0.0577029434655211\\
0.6	0.0828717976310822\\
0.7	0.112161302757422\\
0.8	0.145266867125056\\
0.9	0.181850478642161\\
1	0.221556395058428\\
1.1	0.264025041297929\\
1.2	0.308904537084788\\
1.3	0.355859622738178\\
1.4	0.404578027190447\\
1.5	0.45477451499584\\
1.6	0.506192961080672\\
1.7	0.558606846674801\\
1.8	0.611818565110043\\
1.9	0.665657889701508\\
2	0.719979902731249\\
2.1	0.77466262584651\\
2.2	0.829604535436775\\
2.3	0.884722096079807\\
2.4	0.939947402899973\\
2.5	0.995225990023453\\
2.6	1.05051483664977\\
2.7	1.10578058350932\\
2.8	1.16099795939824\\
2.9	1.21614840885882\\
3	1.27121890682139\\
3.1	1.32620094323141\\
3.2	1.38108965962279\\
3.3	1.43588311970035\\
3.4	1.4905816968373\\
3.5	1.54518756267032\\
3.6	1.59970426246778\\
3.7	1.65413636450985\\
3.8	1.70848917225492\\
3.9	1.7627684895173\\
4	1.81698043021357\\
};
\addlegendentry{$\lambda = 1$, FO}

\addplot [color=mycolor4, dashed, line width=1.5pt]
  table[row sep=crcr]{%
  0	0\\
0.1	0.00456122597382702\\
0.2	0.0184546902368696\\
0.3	0.0414358988155411\\
0.4	0.0726334335625282\\
0.5	0.110778197529214\\
0.6	0.154452268542394\\
0.7	0.202289013595224\\
0.8	0.253096480540336\\
0.9	0.305909282555022\\
1	0.359989951365625\\
1.1	0.414802267718387\\
1.2	0.469973701449987\\
1.3	0.525257418324884\\
1.4	0.580498979699121\\
1.5	0.635609453410694\\
1.6	0.690544829811396\\
1.7	0.745290848418651\\
1.8	0.799852131233888\\
1.9	0.854244586127459\\
2	0.908490215106786\\
2.1	0.962613650827724\\
2.2	1.01663991303067\\
2.3	1.07059301303904\\
2.4	1.12449513923671\\
2.5	1.17836623424357\\
2.6	1.23222383103622\\
2.7	1.2860830557406\\
2.8	1.33995673352504\\
2.9	1.39385555423214\\
3	1.44778826854382\\
3.1	1.50176189534391\\
3.2	1.5557819277939\\
3.3	1.60985253036468\\
3.4	1.66397672230572\\
3.5	1.71815654522703\\
3.6	1.77239321393278\\
3.7	1.8266872505989\\
3.8	1.88103860298568\\
3.9	1.93544674773045\\
4	1.9899107799514\\
};
\addlegendentry{$\lambda = 2$, FO}

\addplot [color=mycolor5, dashed, line width=1.5pt]
  table[row sep=crcr]{%
  0	0\\
0.1	0.00922734511843479\\
0.2	0.0363167167812641\\
0.3	0.077226134271197\\
0.4	0.126548240270168\\
0.5	0.179994975682812\\
0.6	0.234986850724993\\
0.7	0.29024948984956\\
0.8	0.345272414905698\\
0.9	0.399926065616944\\
1	0.454245107553393\\
1.1	0.508319956515336\\
1.2	0.562247569618356\\
1.3	0.616111915518112\\
1.4	0.669978366762518\\
1.5	0.723894134271909\\
1.6	0.77789096389695\\
1.7	0.831988361152862\\
1.8	0.88619660696639\\
1.9	0.940519301492842\\
2	0.994955389975698\\
2.1	1.04950071596784\\
2.2	1.10414917839229\\
2.3	1.15889357239162\\
2.4	1.21372618578589\\
2.5	1.26863921104446\\
2.6	1.32362502059948\\
2.7	1.37867634259991\\
2.8	1.43378636529399\\
2.9	1.48894879112134\\
3	1.54415785607696\\
3.1	1.59940832570418\\
3.2	1.65469547591381\\
3.3	1.71001506447868\\
3.4	1.76536329732276\\
3.5	1.82073679246092\\
3.6	1.8761325435283\\
3.7	1.93154788417905\\
3.8	1.98698045416338\\
3.9	2.04242816755951\\
4	2.09788918340491\\
};
\addlegendentry{$\lambda = 4$, FO}

\end{axis}
\end{tikzpicture}%
}
\caption{The expected delay $E(d)$.}\label{fig: expected delay}
%\end{center}
%\vspace{-20pt}
\end{figure}

\subsection{Delay and Traffic Density}
In general, larger traffic density results in larger delay. According to \Cref{fig: zero delay lambda}, the probability of zero delay $P_d(0)$ drops when the traffic density goes up. In FIFO, it drops linearly and reaches zero when the equality in \eqref{eq: convergence FIFO} holds, where
\begin{equation}
\lambda = \frac{(1+r)^2}{2\Delta_dr}\text{.}\label{eq: FIFO convergence condition simplified}
\end{equation}
In FO, $P_d(0)$ drops with decreasing rate. 
%When the traffic density is low, $P_d(0)$ in FO is tangent to that in FIFO. When the traffic density is high, $P_d(0)$ drops slower than in FIFO. %This is because that the passing order is adaptable. 
According to \Cref{fig: expected delay lambda}, the expected delay $E(d)$ grows with the traffic density $\lambda$. In FIFO, it grows exponentially with $\lambda$, and goes to infinity when $\lambda$ approaches \eqref{eq: FIFO convergence condition simplified}. In FO, it grows with decreasing rate when $\lambda$ increases. 
\Cref{fig: delay distribution lambda} illustrates the distribution of delay for $\lambda \in\{0.1,0.5,1,2,4\}$, $\Delta_d=2$, and $r=0.5$. The distribution does not converge for $\lambda>1.125$ in FIFO, while it always converge in FO. It is easy to verify that the necessary condition \eqref{eq: convergence FO} is always satisfied when $\Delta_s=0$. %These relationships allows more precise traffic prediction.

\subsection{Delay and Density Ratio}
In general, there are more delays when the traffic is more balanced. According to \Cref{fig: zero delay ratio}, $P_d(0)$ drops with decreasing rate when the density ratio approaches $1$. In FIFO, it reaches zero when \eqref{eq: FIFO convergence condition simplified} holds. In FO, $P_d(0)$ is relatively constant for $r>0.5$. 
According to \Cref{fig: expected delay ratio}, the expected delay $E(d)$ grows with respect to the density ratio $r$. In FIFO, the expected delay grows exponentially with $r$ when there is a solution for $r\leq 1$ in \eqref{eq: FIFO convergence condition simplified} for fixed $\lambda$ and $\Delta_d$, e.g.,
\begin{equation}
\lambda\Delta_d\geq \min_{r\in(0,1]}\frac{(1+r)^2}{2r} = 2\text{.}
\end{equation}
The expected delay grows with decreasing rate when there is no solution for $r\leq 1$ in \eqref{eq: FIFO convergence condition simplified}, i.e., $\lambda\Delta_d<2$. In FO, the expected delay grows in decreasing rate when $r$ approaches~$1$. When $\lambda\Delta_d$ is small, the expected delay in FIFO is close to the expected delay in FO. 
%\Cref{fig: delay distribution ratio} illustrates the distribution of delay for $r \in\{0.1,0.2,0.4,0.6,1\}$, $\Delta_d=2$, and $\lambda=1$. The distribution does not change much for large $r$ in the FO case.

\begin{figure}[t]
%\begin{center}
\subfloat[Fix $\Delta_d =\SI{2}{\second}$ and $r=0.5$.\label{fig: delay distribution lambda}]{
\input{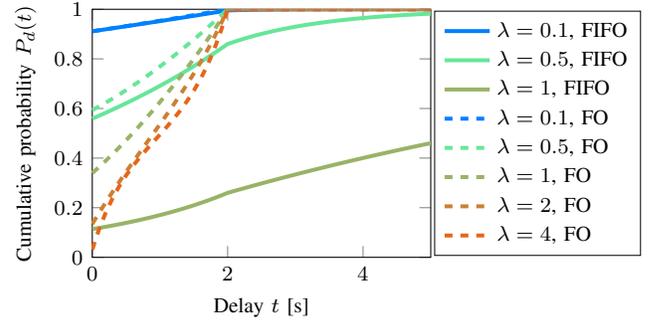}}\\
\subfloat[Fix $\Delta_d =\SI{2}{\second}$ and $\lambda=\SI{1}{\per\second}$.\label{fig: delay distribution ratio}]{
\input{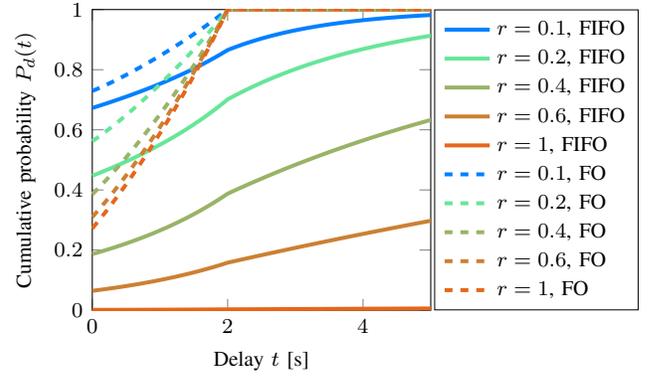}}\\
\subfloat[Fix $\lambda$=\SI{1}{\per\second} and $r=0.5$.\label{fig: delay distribution delta}]{
\input{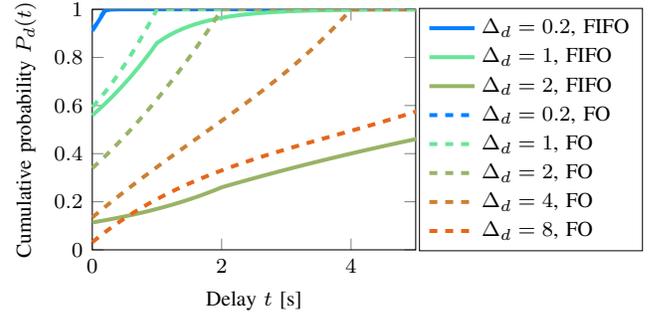}}
\caption{Distribution of steady state vehicle delay $P_d(t)$.}
%\input{FIFO-Uneven-Multiple.tex}
%\caption{Steady state vehicle delays in unbalanced traffic for $\lambda=0.8s^{-1}$, $\Delta_d = 2s$, and $\Delta_s = 0s$.}
\label{fig: distribution}
%\end{center}
\end{figure}

\subsection{Delay and Passing Order}

For all scenarios in \Cref{fig: zero delay}, \Cref{fig: expected delay}, and \Cref{fig: distribution}, FO results in smaller delay than FIFO. The advantage of FO is due to the fact that the passing order is adaptable to real time scenarios. They have similar performances when either $\lambda$, $\Delta_d$, or $r$ is small. In those cases, the order determined by FO is close to the order in FIFO. Moreover, it is worth noting that the delay distribution in \cref{fig: distribution} is not computed for a single vehicle, but for all vehicles on average. Such average delay does not exceed $\Delta_d$ in FO, but it is possible for individual vehicles to have delay greater than $\Delta_d$. Though FO is efficient in the sense that it minimizes delay, it sacrifices fairness by not obeying the passing order determined by the desired passing time. As a consequence, certain vehicles may experience larger delay compared to that in the FIFO case. The tradeoff between fairness and efficiency in different policies will be studied in the future.

%\begin{figure}[t]
%\begin{center}
%\input{Compare-FIFO-FO.tex}
%\caption{Comparison between the lane delay under FO and FIFO. $\lambda_1=0.3s^{-1}$, $\lambda_2=0.5s^{-1}$, $\Delta_d = 2s$, and $\Delta_s = 0s$.}
%\label{default}
%\end{center}
%\end{figure}

\subsection{Delay and Temporal Gap}
In general, a larger temporal gap results in larger delay. According to \Cref{fig: zero delay delta}, $P_d(0)$ drops when the temporal gap $\Delta_d$ increases. In FIFO, it drops linearly and reaches zero when the equality in \eqref{eq: FIFO convergence condition simplified} holds. In FIFO, it drops with decreasing rate. According to \Cref{fig: expected delay delta}, the expected delay $E(d)$ grows with respect to the temporal gap $\Delta_d$. In FIFO, the expected delay grows exponentially. In FO, it eventually reaches a constant growth rate. The temporal gap is a design parameter in vehicle policies, which is affected by the uncertainty in perceptions. When there are larger uncertainties in perception, in order to stay safe, vehicles tend to maintain larger gaps to other vehicles. The trade-off between safety and efficiency under imperfect perception will be studied in the future.

\section{Conclusion\label{sec: conclusion}}
This paper presented a new approach to perform delay analysis for unmanaged intersections in an event-driven stochastic model. The model considered the traffic delay at an intersection as an event-driven stochastic process, whose dynamics encoded equilibria resulted from microscopic multi-vehicle interactions. With the model, the distribution of delay can be obtained through either direct analysis or event-driven simulation. %Both methods were more efficient than conventional traffic simulation, but still captured enough microscopic details compared to conventional macroscopic traffic models.
In particular, this paper performed detailed analyses for a two-lane intersection under two different classes of policies corresponding to two different passing orders. The convergence of the distribution of delay and the steady state delay were derived through direct analysis. The relationships between traffic delay and multiple factors such as traffic flow density, unevenness of traffic flows, temporal gaps between two consecutive vehicles, and the passing order were discussed. 
In the future, such analysis will be extended to more complex vehicle policies, more complex road topologies, multiple intersections, and heterogeneous traffic scenarios.

\bibliographystyle{ieeetr}
%\bibliography{traffic}

\begin{thebibliography}{10}

\bibitem{mathew2014signalized}
T.~V. Mathew, ``Signalized intersection delay models,'' {\em Lecture notes in
  Traffic Engineering and Management}, 2014.

\bibitem{xi2015approach}
J.~Xi, W.~Li, S.~Wang, and C.~Wang, ``An approach to an intersection traffic
  delay study based on shift-share analysis,'' {\em Information}, vol.~6,
  no.~2, pp.~246--257, 2015.

\bibitem{jiang2005traffic}
Y.~Jiang, S.~Li, and K.~Q. Zhu, ``Traffic delay studies at signalized
  intersections with global positioning system devices,'' {\em ITE Journal},
  vol.~75, no.~8, pp.~30--39, 2005.

\bibitem{vanmiddlesworth2008replacing}
M.~VanMiddlesworth, K.~Dresner, and P.~Stone, ``Replacing the stop sign:
  Unmanaged intersection control for autonomous vehicles,'' in {\em
  International Joint Conference on Autonomous Agents and Multiagent Systems},
  vol.~3, pp.~1413--1416, IFAAMS, 2008.

\bibitem{savic2017distributed}
V.~Savic, E.~M. Schiller, and M.~Papatriantafilou, ``Distributed algorithm for
  collision avoidance at road intersections in the presence of communication
  failures,'' in {\em Intelligent Vehicles Symposium (IV)}, pp.~1005--1012,
  IEEE, 2017.

\bibitem{azimi2013reliable}
S.~Azimi, G.~Bhatia, R.~Rajkumar, and P.~Mudalige, ``Reliable intersection
  protocols using vehicular networks,'' in {\em International Conference on
  Cyber-Physical Systems}, ICCPS '13, pp.~1--10, ACM, 2013.

\bibitem{liu2017distributed}
C.~Liu, C.~W. Lin, S.~Shiraishi, and M.~Tomizuka, ``Distributed conflict
  resolution for connected autonomous vehicles,'' {\em IEEE Transactions on
  Intelligent Vehicles}, vol.~3, no.~1, pp.~18--29, 2018.

\bibitem{gora2016traffic}
P.~Gora and I.~R{\"u}b, ``Traffic models for self-driving connected cars,''
  {\em Transportation Research Procedia}, vol.~14, pp.~2207 -- 2216, 2016.
\newblock Transport Research Arena TRA2016.

\bibitem{treiber2010open}
M.~Treiber and A.~Kesting, ``An open-source microscopic traffic simulator,''
  {\em IEEE Intelligent Transportation Systems Magazine}, vol.~2, no.~3,
  pp.~6--13, 2010.

\bibitem{barcelo1999modelling}
J.~Barcel{\'o}, J.~Casas, J.~L. Ferrer, and D.~Garc{\'\i}a, {\em Modelling
  Advanced Transport Telematic Applications with Microscopic Simulators: The
  Case of {AIMSUN2}}, pp.~205--221.
\newblock Springer, 1999.

\bibitem{Fellendorf2010}
M.~Fellendorf and P.~Vortisch, {\em Microscopic Traffic Flow Simulator VISSIM},
  pp.~63--93.
\newblock Springer, 2010.

\bibitem{hoogendoorn2001state}
S.~P. Hoogendoorn and P.~H.~L. Bovy, ``State-of-the-art of vehicular traffic
  flow modelling,'' {\em Proceedings of the Institution of Mechanical
  Engineers, Part I: Journal of Systems and Control Engineering}, vol.~215,
  no.~4, pp.~283--303, 2001.

\bibitem{CORTHOUT2012343}
R.~Corthout, G.~Fl{\"o}tter{\"o}d, F.~Viti, and C.~M. Tamp{\`e}re, ``Non-unique
  flows in macroscopic first-order intersection models,'' {\em Transportation
  Research Part B: Methodological}, vol.~46, no.~3, pp.~343 -- 359, 2012.

\bibitem{flotterod2011operational}
G.~Flotterod and J.~Rohde, ``Operational macroscopic modeling of complex urban
  road intersections,'' {\em Transportation Research Part B: Methodological},
  vol.~45, no.~6, pp.~903 -- 922, 2011.

\bibitem{liu2018analytically}
C.~{Liu} and M.~J. {Kochenderfer}, ``{Analytically Modeling Unmanaged
  Intersections with Microscopic Vehicle Interactions},'' {\em
  arXiv:1804.04746}, Apr. 2018.

\bibitem{altche2016time}
F.~Altch{\'e}, X.~Qian, and A.~de~La~Fortelle, ``Time-optimal coordination of
  mobile robots along specified paths,'' in {\em International Conference on
  Intelligent Robots and Systems (IROS)}, pp.~5020--5026, IEEE, 2016.

\bibitem{qian2017autonomous}
X.~Qian, F.~Altch{\'e}, J.~Gr{\'e}goire, and A.~de~La~Fortelle, ``Autonomous
  intersection management systems: criteria, implementation and evaluation,''
  {\em IET Intelligent Transport Systems}, vol.~11, no.~3, pp.~182--189, 2017.

\bibitem{gardiner2009stochastic}
C.~Gardiner, {\em Stochastic methods}, vol.~4.
\newblock Springer, 2009.

\bibitem{dde}
W.~Michiels and S.-I. Niculescu, {\em Spectral Properties of Linear Time-Delay
  Systems}, ch.~1, pp.~3--31.
\newblock SIAM, 2007.

\end{thebibliography}

%\printbibliography
%\vspace{-20pt}
%\begin{IEEEbiography}[{\includegraphics[width=1in,height=1.25in,clip,keepaspectratio]{fig/Changliu.jpg}}]{Changliu Liu} received the Ph.D. degree in mechanical engineering from University of California, Berkeley, U.S.A., in 2017. She is a postdoctoral fellow at Stanford. Her research interests include robotics, motion planning,  and optimization. 
%\end{IEEEbiography}
%\vspace{-20pt}
%\begin{IEEEbiography}[{\includegraphics[width=1in,height=1.25in,clip,keepaspectratio]{fig/Mykel.jpg}}]{Mykel J. Kochenderfer} received the B.S. and M.S.
%	degrees in computer science from Stanford University in 2003 and the Ph.D. degree
%	from The University of Edinburgh in 2006.
%	He was with MIT Lincoln Laboratory, where he
%	worked on airspace modeling and aircraft collision
%	avoidance. Since 2013, he has been an Assistant Professor
%	of aeronautics and astronautics with Stanford
%	University. His research interests include advanced
%	algorithms and analytical methods for the design of
%	robust decision-making systems.
%\end{IEEEbiography}

\end{document}